\documentclass[pra,aps,superscriptaddress,twocolumn,letter,nopacs,nofootinbib,longbibliography]{revtex4-2}

\usepackage[english]{babel}
\usepackage[utf8x]{inputenc}
\usepackage[T1]{fontenc}

\usepackage{amsmath}
\usepackage{amsfonts}
\usepackage{amssymb,amscd}
\usepackage{graphicx}
\usepackage[colorinlistoftodos]{todonotes}
\usepackage[colorlinks=true, allcolors=blue]{hyperref}
\usepackage{braket}
\usepackage{comment}
\usepackage{float}
\usepackage{bbm,amsfonts}
\usepackage[shortlabels]{enumitem}
\usepackage{hyperref}
\usepackage{amsthm}
\usepackage{xcolor}
\usepackage[title]{appendix}
\usepackage[capitalise, noabbrev]{cleveref}

\usepackage{tikz, pgf, pgfplots, pgfplotstable}
\usepgfplotslibrary{groupplots}
\pgfplotsset{compat=1.15}
\usepgfplotslibrary{colorbrewer}
\usetikzlibrary{plotmarks}
\usetikzlibrary{arrows}

\hypersetup{
	colorlinks=true,  
	linkcolor=blue,   
	citecolor=blue,   
	urlcolor=blue     
}

\def\B{ {\mathcal B} }
\def\C{ {\mathcal C} }

\def\F{ {\mathcal F} }

\def\L{ {\mathcal L} }

\def\U{ {\mathcal U} }

\newcommand{\tra}[1]{\mathrm{tr}\left( #1 \right)}

\def\>{\rangle}
\def\<{\langle}

\newcommand{\ketbra}[2]{\ensuremath{\left|#1\right\rangle\!\!\left\langle#2\right|}}

\newcommand{\matrixel}[3]{\ensuremath{\left\langle #1 \vphantom{#2#3} \right| #2 \left| #3 \vphantom{#1#2} \right\rangle}}
\newcommand{\iden}{\mathbb{I}} 

\def\non{ \nonumber\\}

\newtheorem{theorem}{Theorem}

\newtheorem{lemma}[theorem]{Lemma}
\newtheorem{definition}[theorem]{Definition}

\newtheorem{corollary}[theorem]{Corollary}
\newtheorem{conjecture}[theorem]{Conjecture}

\begin{document}

	
\title{Algebraic and geometric structures inside the Birkhoff polytope}

\author{Grzegorz Rajchel-Mieldzioć}
\affiliation{Center for Theoretical Physics, Polish Academy of Sciences, 02-668 Warszawa, Poland}

\author{Kamil Korzekwa}
\affiliation{Faculty of Physics, Astronomy and Applied Computer Science, Jagiellonian University, 30-348 Krak\'{o}w, Poland}

\author{Zbigniew Pucha{\l}a}
\affiliation{Faculty of Physics, Astronomy and Applied Computer Science, Jagiellonian University, 30-348 Krak\'{o}w, Poland}
\affiliation{Institute of Theoretical and Applied Informatics, Polish Academy of Sciences, 44-100 Gliwice, Poland}

\author{Karol {\.Z}yczkowski}
\affiliation{Center for Theoretical Physics, Polish Academy of Sciences, 02-668 Warszawa, Poland}
\affiliation{Faculty of Physics, Astronomy and Applied Computer Science, Jagiellonian University, 30-348 Krak\'{o}w, Poland}

\date{\today}

\begin{abstract}
    The Birkhoff polytope~$\B_d$ consisting of all bistochastic matrices of order~$d$ assists researchers from many areas, including combinatorics, statistical physics and quantum information. Its subset~$\U_d$ of unistochastic matrices, determined by squared moduli of unitary matrices, is of a particular importance for quantum theory as classical dynamical systems described by unistochastic transition matrices can be quantised. In order to investigate the problem of unistochasticity we introduce the set~$\L_d$ of bracelet matrices that forms a subset of~$\B_d$, but a superset of~$\U_d$. We prove that for every dimension~$d$ this set contains the set of factorisable bistochastic matrices~$\F_d$ and is closed under matrix multiplication by elements of~$\F_d$. Moreover, we prove that both~$\L_d$ and~$\F_d$ are star-shaped with respect to the flat matrix. We also analyse the set of~$d\times d$~unistochastic matrices arising from circulant unitary matrices, and show that their spectra lie inside~$d$-hypocycloids on the complex plane. Finally, applying our results to small dimensions, we fully characterise the set of circulant unistochastic matrices of order~$d\leq 4$, and prove that such matrices form a monoid for $d=3$.
\end{abstract}

\maketitle


\section{Introduction}
\label{sec:intro}

Discrete-time dynamics of a finite-dimensional classical system, whose state is described by a $d$-dimensional probability vector, can be represented by a left stochastic matrix of size $d$. The entries of such a matrix are all non-negative and are summing to unity in each column. By imposing a symmetric constraint on rows one obtains the set $\B_d$ of $d$-dimensional \emph{bistochastic} matrices, 
also called doubly-stochastic. Through the Birkhoff's theorem~\cite{bhatia2013matrix}, the set $\B_d$ has a clear geometric interpretation, as it is given by the convex hull of all permutation matrices of order $d$. The set $\B_d$ forms a polytope in ${\mathbbm R}^{(d-1)^2}$, often called the \emph{Birkhoff polytope}, with vertices given by $d\ \! \! !$ permutation matrices. Its centre is occupied by the flat matrix $W_d$, with all entries equal to $1/d$, which is also called the van der Waerden matrix due to his conjecture on minimal permanent of a bistochastic matrix~\cite{van1926aufgabe}. From the algebraic perspective, the set $\B_d$ forms a \emph{monoid}, as it is closed under (associative) matrix multiplication and contains the identity matrix $\iden_d$.

In this work we will investigate the connection between $\B_d$ and the set of $d\times d$ unitary matrices that describe the closed dynamics of a $d$-dimensional quantum system. As the rows and columns of a unitary matrix $U$ are normalised to unity, every matrix with entries $|U_{jk}|^2$ is bistochastic. However, the converse is not true. This is due to the fact that the rows and columns of unitary matrices are not only normalised, but they also need to be mutually orthogonal, which results in additional constraints. A bistochastic matrix that stems from some unitary matrix is called \emph{unistochastic}~\cite{Be04}.

The problem of characterising the set $\U_d$ of unistochastic matrices of size $d$ is related to the issue of quantization of classical dynamical systems determined by a given bistochastic transition matrix~\cite{PKZ01,PTZ03}. Alternatively, unistochasticity is also linked
to the problem of finding all discrete quantum walks on graphs with $d$ vertices. Recall that in a classical random walk on the line the walker starts at a node $n=0$ and at each time step moves with probability $p_{-}$ to the left (to the node $n-1$), with probability $p_+$ to the right (to the node $n+1$), and with probability $p_0$ does not move (stays at the node $n$). The straightforward quantum equivalent of such a process would be given by a unitary transformation that transforms the basis states $\{\ket{n}\}$ with $n\in\mathbb{Z}$ as follows~\cite{ambainis2003quantum}
\begin{equation}
	\ket{n}\rightarrow c_- \ket{n-1} + c_0 \ket{n} + c_+ \ket{n+1},
\end{equation}
with $|c_i|^2=p_i$, where $i\in\{\pm,0\}$. However, as proven in Ref.~\cite{meyer1996quantum}, no unitary process can induce such a transformation beyond the trivial cases (i.e., one $p_i=1$ and the remaining ones equal to 0). Therefore, most researchers focus on a less straightforward definition of a quantum random walk that was first proposed by the authors of Ref.~\cite{aharonov1993quantum}. It is based on introducing the additional coin degree of freedom, so that walker's states are given by $\ket{n,k}$, where $n$ corresponds to the position of the walker and $k$ to the state of the coin. Then, at each step of the quantum random walk there is a ``coin flip'' (unitary transformation of the coin), followed by the ``shift'' (change of the position based on the state of the coin). 

Not being discouraged by the fact that a natural quantum generalisation of a random walk is impossible on an infinite line, one can ask when it is possible. Focusing on random walks on finite graphs, the set of states is not given by all integer numbers, but rather by a set $\{1,\dots,d\}$. A general walker then moves from a node $k$ to a node $j$ with transition probability $T_{jk}$, and so its behaviour is fully specified by a stochastic transition matrix $T$. Now, the straightforward quantum equivalent of this process is given by a unitary transformation $U$ that transforms the basis states $\{\ket{k}\}$ with $k\in\{1,\dots,d\}$ as follows
\begin{equation}
U\ket{k} = \sum_{j=1}^{d} U_{jk} \ket{j}, \quad \mathrm{with}\quad|U_{jk}|^2=T_{jk}.
\end{equation}
Therefore, given a classical random walk on a graph with $d$ nodes and described by a stochastic matrix $T$, its quantum equivalent exists if and only there exists a unitary matrix $U$ satisfying the above equation. It is thus clear that the necessary and sufficient condition for the existence of a quantum random walk corresponding to a classical walk $T$ is that the transition matrix is unistochastic.

Beyond the mere classification of random walks into the ones for which quantum realisations do or do not exist, investigating the problem of unistochasticity can give us insight into the nature of randomness~\cite{korzekwa2018coherifying}. Classically, random processes are necessarily irreversible; but unitary quantum processes are fully deterministic and reversible despite the fact that they can lead to random measurement outcomes. Thus, these random walks that correspond to unistochastic matrices can be generated in nature completely deterministically, despite the apparent randomness they induce in a given step. As an illustrative example, consider an analogue of the random walk on the line described above, but realised on a cyclic graph with four vertices (i.e., moving to the right of node~4 brings the walker to node~1). In Sec.~\ref{sec:uni}, we prove that every matrix corresponding to such a walk with $p_+=p_-$ is unistochastic, and so there always exists a quantum walk of that type. Now, if we measure the walker's position after every step (effectively rendering its dynamics classical), it quickly spreads over the whole graph, i.e., independently of the starting position its final position quickly becomes uniformly distributed over all four nodes. However, if we do not observe the walker and let it evolve over $N$ steps before measuring its position, we will never find the walker at node 3~if it started at node~1 (and vice verse), as well as we will never find it at node~2 if it started at node~4 (and vice versa). This kind of behaviour is based on the reversibility of the underlying unitary dynamics and thus can only be observed for unistochastic processes.

Additional incentive motivating our study of unistochastic matrices comes from the deep connection to particle physics, more concretely from the violation of charge conjugation parity symmetry (CP symmetry). Such violation was first observed experimentally and later explained by introducing the third family of quarks. Even though mixing of quarks is described by a unitary Cabibbo-Kobayashi-Maskawa (CKM) matrix~\cite{BS09}, what we access experimentally is a bistochastic matrix built of probabilities~\cite{Di06}. Indeed, existence of only two families of quarks would correspond to CKM matrices of dimension 2 -- but all bistochastic matrices of this dimension are also unistochastic, and that could not lead to CP violation. Thus, the violation is explained by the introduction of dimension 3 unitary CKM matrices, leading to 3 families of quarks~\cite{jarlskog}. Notwithstanding, it is of theoretical importance to study even higher-dimensional unistochastic sets, if only to explore possibilities of more quark families. Note that higher-dimensional problems are interesting also from the point of view of scattering theory~\cite{nuyts_1974}. 

Moreover, while studying quantum foundations researchers often consider the problem of transition probabilities, which leads to bistochastic matrices. In order to introduce group structure, it is preferable to study the unistochastic subset because of its natural group arrangement~\cite{lande, rovelli}. Finally, as recently observed in Ref.~\cite{korzekwa2020quantum}, unistochastic matrices are interesting from the perspective of quantum memory advantages, as they can be generated by memoryless quantum dynamics, even though classically they may be non-Markovian. All things considered, the unistochasticity problem is deeply embedded in a wide variety of topics in physics~\cite{Be04}.

The main aim of this work is to contribute to our understanding of the structure of the set $\U_d$ of unistochastic matrices and to analyse possible dynamics induced by this set, motivated by applications in quantum physics. To achieve this, we introduce the set $\L_d$ of \emph{bracelet matrices} that includes the analysed set $\U_d$ and approximates it from outside. In fact, in certain low-dimensional cases these two sets coincide, and so the properties of $\L_d$ directly translate into the properties of $\U_d$. Based on numerical evidence, we conjecture that the set of bracelet matrices forms monoid and prove a slightly weaker algebraic property. Namely, we show that $\L_d$ is closed under matrix multiplication by elements of its subset $\F_d$ of factorisable matrices~\cite{poon1987inclusion}. Moreover, we prove that both $\L_d$ and $\F_d$ are star-shaped with respect to the Van der Waerden matrix $W_d$, contributing to the analysis of star-shapedness of $\U_d$ for higher dimensions (for $d=3$ this property has been observed previously in Ref~\cite{bengtsson2005birkhoff}). We also investigate \emph{doubly circulant} unistochastic matrices (i.e., these unistochastic matrices that arise from circulant unitary dynamics) and characterise their spectra. Finally, our general results allow us to solve the problem of unistochasticity for circulant matrices of order $d=4$, and demonstrate that for $d=3$ the set of circulant unistochastic matrices forms a monoid.
 
The paper is organised as follows. First, in Sec.~\ref{sec:preliminaries}, we provide a mathematical background for our work, review some earlier results and introduce the notion of bracelet matrices. Then, in Sec.~\ref{sec:bracelet}, we characterise the algebraic and geometric structure of the set $\L_d$ of bracelet matrices. In Sec.~\ref{sec:uni} we employ the general results on bracelet matrices to investigate the structure of unistochastic matrices. Finally, Sec.~\ref{sec:conclusions} contains conclusions and outlook for future work.


\section{Mathematical preliminaries}
\label{sec:preliminaries}

In this section we provide rigorous mathematical definitions of the investigated subsets of the Birkhoff polytope and describe their known properties. We illustrate the relation between some of these subsets in Fig.~\ref{fig:sets}.

\begin{definition}[Bistochastic matrices]
	The set of $d\times d$ bistochastic matrices $\B_d$ consists of all matrices $B$ whose entries are non-negative and satisfy
	\begin{equation}
		\sum_{j=1}^d B_{jk}=\sum_{k=1}^d B_{jk}=1.
	\end{equation}
\end{definition}

The Birkhoff polytope $\B_d$ forms a convex body in ${\mathbbm R}^{(d-1)^2}$, and its volume with respect to the Euclidean measure is explicitly known for small dimensions~\cite{CR99, DLLY09} and asymptotically for large $d$~\cite{CMK07}. To generate random bistochastic matrices that cover the entire set $\B_d$ one can use the Sinkhorn algorithm~\cite{Si64}, but the measure induced in this way is not uniform in the polytope~\cite{CSBZ09}.

The simplest bistochastic matrices, known as \emph{elementary bistochastic matrices}~\cite{marcus1984products}, act non-trivially on at most a two-dimensional subspace.
\begin{definition}[Elementary bistochastic matrices]
	A~\mbox{$d\times d$} bistochastic matrix $B$ is elementary if $B_{kk}=1$ for at least $d-2$ indices $k$.
\end{definition}
We note that these matrices, which contain at most $d+2$ non-zero entries, are also known in the literature under alternative names of $T$-transforms~\cite{bhatia2013matrix} and pinching matrices~\cite{poon1987inclusion}. Of special interest for us will be the set of matrices that can be decomposed as a product of elementary matrices. 
\begin{definition}[Factorisable bistochastic matrices]
	The closed set of $d\times d$ factorisable bistochastic matrices $\F_d$ consists of all matrices $B$ that can be expressed as a product of elementary bistochastic matrices, 
	\begin{equation}
		B=\prod_{k=1}^n A^{(k)},
	\end{equation}
	with $A^{(k)}$ being elementary and an integer $n$ being possibly infinite.
\end{definition}
The relation of factorisable matrices to the problem of unistochasticity was investigated in Ref.~\cite{poon1987inclusion}, while more recently such matrices were employed in the studies of quantum thermodynamics~\cite{lostaglio2018elementary} and Markovian quantum dynamics~\cite{korzekwa2020quantum}.

The main interest of our research lies in the set $\U_d$ of unistochastic matrices defined as follows.

\begin{definition}[Unistochastic matrices]
	\label{def:uni}
	The set of \mbox{$d\times d$} unistochastic matrices $\U_d$ consists of all bistochastic matrices $B$ whose entries satisfy $B_{jk} = |U_{jk}|^2$ for some unitary matrix $U$.
\end{definition}

\begin{figure}
	\begin{tikzpicture}
	\node (myfirstpic) at (0,0) {\includegraphics[width=0.8\columnwidth]{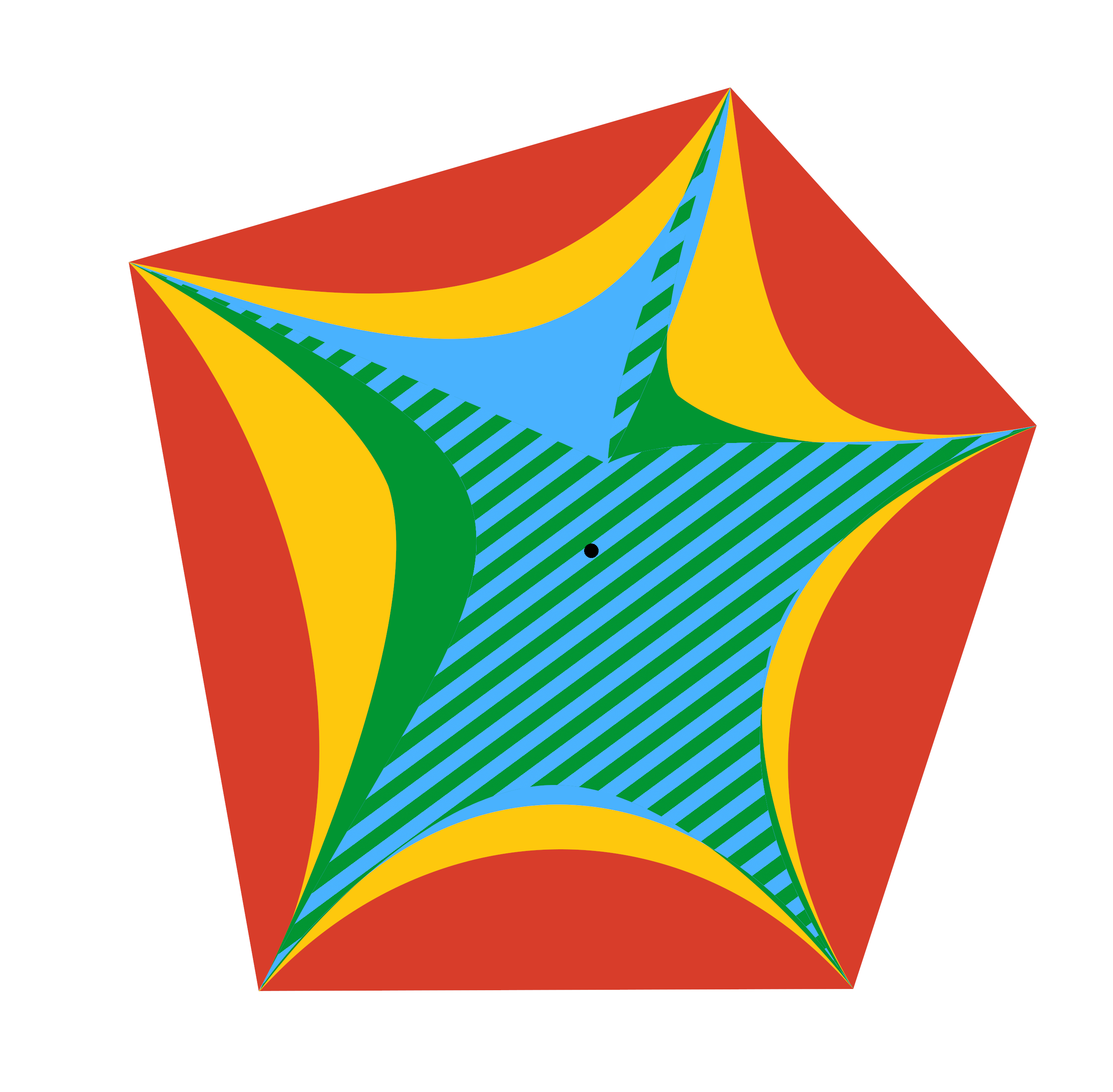}};		
	\node at (0.12\columnwidth,0.15\columnwidth) {\color{black}$\L_d$};
	\node at (0.25\columnwidth,-0.34\columnwidth) {\color{black}$\iden_d$};
	\node at (0.38\columnwidth,0.1\columnwidth) {\color{black}$\Pi_d$};
	\node at (0\columnwidth,-0.27\columnwidth) {\color{black}$\B_d$};		
	\node at (-0.085\columnwidth,-0.0\columnwidth) {\color{black}$\U_d$};			
	\node at (0.06\columnwidth,0.0\columnwidth) {\color{black}$W_d$};	
	\node at (-0.0\columnwidth,-0.13\columnwidth) {\color{black}$\F_d$};	
	\end{tikzpicture}
	\caption{\label{fig:sets} \textbf{Families of investigated transition matrices.} The inclusion relations for generic dimension $d>3$ are given by $\U_d \subset \L_d \subset \B_d$ with $\F_d \subset \L_d$, but with non-comparable sets $\U_d$ and $\F_d$. Note that for some dimensions, e.g. \mbox{$d=4$}, the flat matrix $W_d$ lies at the boundary of $U_d$~\cite{bengtsson2005birkhoff}, and that the investigated sets of circulant and doubly circulant unistochastic matrices, $\C_d$ and $\C_d^2$, are not presented in the picture.}
\end{figure}

The above definition can be treated as a mapping from the $d^2$ dimensional set of unitary matrices down to the $(d-1)^2$ dimensional set $\B_d$ of bistochastic matrices~\cite{Ka08}. It is often convenient to write a unistochastic matrix $B_U\in \U_d$ as a Hadamard (entry-wise) product, $B_U=U\circ {\bar U}$.
If there exists an orthogonal matrix $O$ such that $B=O\circ {O}$, the matrix $B$ is called {\sl orthostochastic} and the set of such matrices will be denoted by ${\cal O}_d$, even though in some mathematical works~\cite{poon1987inclusion,AYC91,Na96} this name refers also to unistochastic matrices.

For $d=2$ any bistochastic matrix is unistochastic and also orthostochastic, so that these three sets are equal, $\U_2=\B_2={\cal O}_2$. To illustrate the fact that for larger dimensions this is not the case, consider the following bistochastic matrix of order $d=3$,
\begin{equation}
	\label{eq:Q_matrix}
	Q = \frac{1}{2}\begin{pmatrix}
		0&1 &1 \\
		1 &0 &1 \\
		1 &1 &0 
	\end{pmatrix}.
\end{equation}
The corresponding unitary matrix must then be of the form
\begin{equation}
	U_3=\frac{1}{\sqrt{2}}\begin{pmatrix}
		0&e^{i\phi_1}&e^{i\phi_2} \\
		e^{i\phi_3} &0&e^{i\phi_4} \\
		e^{i\phi_5} &e^{i\phi_6}&0
	\end{pmatrix}.
\end{equation}
Orthogonality of the first and second rows requires the vanishing of $e^{i(\phi_2-\phi_4)}$, which is obviously not possible for any angles $\phi_i$. Therefore, $Q$ can serve as a simple example of a bistochastic matrix, which is not unistochastic. In the case $d=3$ it is known that the Euclidean volume of the set $\U_3$ occupies $8\pi^2/105\approx 0.752$ of the volume of the Birkhoff polytope $\B_3 \subset {\mathbbm R}^4$~\cite{DZ09}. More generally, for every dimension $d \ge 3$ the set $\U_d$ forms a proper subset of $\B_d$. Moreover, geometry of $\U_d$ is more involved than that of $\B_d$, as the set $\U_d$ is not convex. This can also be seen by the above example: simply note that all permutations are unistochastic and that $Q$ is a convex combination of two permutations. Note, however, that it was shown that convex combinations of a particular class of {\sl complementary} permutation matrices are unistochastic~\cite{AYC91}.
 
While analysing the structure of the set $\U_d$, we will also focus on a particular family of circulant unistochastic matrices that correspond to translation-invariant random processes.

\begin{definition}[Circulant matrix]
	The set of \mbox{$d\times d$} circulant matrices $\C_d$ consists of all matrices $B$ whose entries satisfy~\cite{Davis79} 
	\begin{equation}
		B_{j+l,k}=B_{j,k-l},
	\end{equation} 
	where the addition and subtraction is modulo $d$. In other words, the $k$-th row vector is given by the first row vector translated to the right by $(k-1)$.
\end{definition}

As a bistochastic matrix belongs to the $(d-1)^2$-dimensional convex hull of $d!$ permutation matrices, a circulant bistochastic matrix belongs to its $(d-1)$-dimensional section determined by the convex hull of $d$ cyclic permutations matrices. Hence, any circulant bistochastic matrix $C$ can be expressed as a convex combination of powers of the cyclic permutation matrix $\Pi_d$,
\begin{equation}
	\label{ Bcyclic}
	C = \sum_{j=1}^{d}   \alpha_j  \Pi_d^j,\qquad \Pi_d=\left(%
	\begin{array}{ccccc}
		0&1&0&\dots&0\\
		0&0&1&\dots&0\\
		\vdots&\vdots&\vdots&\ddots&1\\
		1&0&0&\dots&0
	\end{array}\right).
\end{equation}
Circulant bistochastic matrices define an interesting class of classical Markov chains, and can also be used to construct circulant completely positive maps and circulant quantum Markov semigroups~\cite{BSQ13,BSC14,BSCQ19}, while circulant unistochastic matrices were analysed in Ref.~\cite{Sm15}.

Note that the fact that a given unistochastic matrix $B$ is circulant does not necessarily mean that there exists an underlying unitary matrix $U$ that is circulant. Of course, the amplitudes of such a unitary matrix must be circulant, but it is not clear whether the phases can be chosen in a way that makes $U$ circulant. Therefore, we introduce the following notion of a doubly circulant unistochastic matrix.
\begin{definition}[Doubly circulant unistochastic matrices]
	The set of \mbox{$d\times d$} doubly circulant unistochastic matrices $\C^2_d$ consists of these circulant unistochastic matrices for which there exists the corresponding unitary matrix that is circulant. 
\end{definition}

Finally, we also introduce and study the properties of the set of \emph{bracelet matrices} $\L_d$ defined as follows.
\begin{definition}[Bracelet matrices]
	\label{def:bracelet}
	The set of \mbox{$d\times d$} bracelet matrices $\L_d$ consists of these bistochastic matrices $B$ whose entries satisfy:
	\begin{subequations}
		\begin{align}
		\label{eq:bracelet1}
		2\max_j \sqrt{B_{jk}B_{jl}}&\leq \sum_{j=1}^d \sqrt{B_{jk}B_{jl}},\\
		\label{eq:bracelet2}
		2\max_l \sqrt{B_{jl}B_{kl}}&\leq \sum_{l=1}^d \sqrt{B_{jl}B_{kl}},
		\end{align}
	\end{subequations}	
	which are called column and row bracelet conditions.
\end{definition}
As we explain below, the name refers to the fact that when inequalities from Eqs.~\eqref{eq:bracelet1}-\eqref{eq:bracelet2} are satisfied, a certain sequence of segments with specified lengths can be closed to form a polygon (a ``bracelet'') in the complex plane. The main reason to investigate bracelet matrices comes from the fact that they form a superset of unistochastic matrices, and thus we can use them as a tool to study the properties of $\U_d$. To see that every matrix $B\in\U_d$ also belongs to $\L_d$, recall that pairs of rows and columns of a unistochastic matrix $B$ are constrained due to the orthogonality of the rows and columns of the corresponding unitary matrix $U$. Noting that the entries of $U$ are of the form
\begin{equation}
	U_{jk}=\sqrt{B_{jk}} e^{i\phi_{jk}},
\end{equation}
the orthogonality conditions read
\begin{subequations}
	\begin{align}
		\sum_{j=1}^d \sqrt{B_{jk}B_{jl}} e^{i(\phi_{jk}-\phi_{jl})}&=0,\\
		\sum_{l=1}^d \sqrt{B_{jl}B_{kl}} e^{i(\phi_{jl}-\phi_{kl})}&=0.
	\end{align}
\end{subequations}
Now, the necessary conditions for the vanishing of the above sums is that the amplitude of the largest element in each sum is smaller than the sum of the remaining amplitudes. Geometrically, this can be understood as a requirement to form a closed polygon in the complex plane out of $d$ segments of length $\sqrt{B_{jk}B_{jl}}$ (or $\sqrt{B_{jl}B_{kl}}$) each. This is clearly impossible if the longest segment is longer than the total length of all the remaining segments. Equations~\eqref{eq:bracelet1}-\eqref{eq:bracelet2} in Definition~\ref{def:bracelet} correspond precisely to these necessary conditions, and so we conclude that if a matrix $B$ is unistochastic, it must satisfy \mbox{$2 {d\choose 2} = d(d-1)$} \emph{bracelet conditions} specified by Eqs.~\eqref{eq:bracelet1}-\eqref{eq:bracelet2}. Thus, whenever $B\in\U_d$, we also have $B\in\L_d$. Finally, it should be added here that the above bracelet conditions, 
Eqs.~\eqref{eq:bracelet1}-\eqref{eq:bracelet2}, that are necessary for unistochasticity, were also called chain-link conditions in earlier mathematical~\cite{AYP79, Na96} and physical~\cite{JS88, PKZ01} literature.


\section{Structure of the set of bracelet matrices}
\label{sec:bracelet}

\subsection{Algebraic properties}

Our crucial result that specifies the algebraic properties of $\L_d$ is given by the following theorem.
\begin{theorem}[Structure of $\L_d$]
	\label{thm:bracelet}
	The set of bracelet matrices $\L_d$ is closed under multiplication by factorisable matrices:
	\begin{equation}
	A\in\F_d\mathrm{~and~}B\in\L_d \Rightarrow AB\in\L_d\mathrm{~and~}BA\in\L_d.
	\end{equation}
	As a result, every factorisable bistochastic matrix is also bracelet:
	\begin{equation}
		B\in \F_d \Rightarrow B\in \L_d.
	\end{equation} 
	
\end{theorem}
Before we present the proof of the above theorem, let us first state a stronger conjecture that is supported by our numerical investigations:
\begin{conjecture}[Monoid structure of $\L_d$]
	\label{conj:bracelet}
	The set of bracelet matrices $\L_d$ forms a monoid, i.e., it is closed under matrix multiplication and contains the identity matrix $\iden_d$.
\end{conjecture}

For the proof of Theorem~\ref{thm:bracelet} we will need three lemmas. These concern pairs of vectors (effectively rows and columns) rather then full bistochastic matrices, but as the bracelet conditions refer to pairs of rows and columns, the result for matrices follows almost straightforwardly from the ones for vectors. First, let us introduce the notion of bracelet conditions for vectors.
\begin{definition}
	We say that two vectors \mbox{$\ket{p}=[p_1,\dots,p_d ]$} and \mbox{$\ket{q}=[q_1,\dots,q_d]$} satisfy the bracelet conditions if and only if
	\begin{equation}
		\sqrt{p_k q_k} \leq \sum_{j\neq k}\sqrt{p_j q_j}
	\end{equation}
	for all $j$. We shall denote this as
	\begin{equation}
		\{\ket{p},\ket{q}\} \in L_d.    
	\end{equation}	
\end{definition}

Now, the first lemma is given as follows.
\begin{lemma}
	\label{lem:columns}
	Let $\ket{p}$ and $\ket{q}$ be vectors with non-negative entries, such that $\{\ket{p},\ket{q}\}\in L_d$. Then, for an elementary bistochastic $E$ we have $\{E\ket{p},E\ket{q}\}\in L_d$.
\end{lemma}
\begin{proof}
	Without loss of generality we may assume \mbox{$E_{11}=E_{22}=t$}, \mbox{$E_{12}=E_{21}=1-t$}, $E_{kk}=1$ for all $k\geq 3$, and all other entries of $E$ being zero. Similarly, we may assume that \mbox{$p_1q_1\geq p_2 q_2$}. Let us denote $\ket{\tilde{p}}=E\ket{p}$ and $\ket{\tilde{q}}=E\ket{q}$. Then, for $k\geq 3$ we have
	\begin{align}
		\sqrt{\tilde{p}_k\tilde{q}_k}=&\sqrt{p_k q_k}\leq \sum_{j\neq k} \sqrt{p_j q_j}
		=\sqrt{p_1q_1}+\sqrt{p_2q_2}\nonumber\\
		&+\!\!\!\!\!\sum_{j\notin\{1,2,k\}} \!\!\!\sqrt{\tilde{p}_j \tilde{q}_j}\leq\sum_{j\neq k} \sqrt{\tilde{p}_j \tilde{q}_j},
	\end{align}
	where the first inequality comes from the assumption that $\ket{p}$ and $\ket{q}$ satisfy the bracelet condition; and the second inequality comes from the monotonicity of the fidelity function (Bhattacharyya coefficient) under stochastic processing by $E$ of the first two entries~\cite{watrous2018theory}.
	
	Thus, in order to prove that $\ket{\tilde{p}}$ and $\ket{\tilde{q}}$ satisfy the bracelet condition, we are left to demonstrate that
	\begin{subequations}
		\begin{align}
			\sqrt{\tilde{p}_1\tilde{q}_1}\leq& \sqrt{\tilde{p}_2\tilde{q}_2}+\sum_{j\geq 3} \sqrt{{p}_j{q}_j},\label{eq:bracelet_lem}\\
			\sqrt{\tilde{p}_2\tilde{q}_2}\leq& \sqrt{\tilde{p}_1\tilde{q}_1}+\sum_{j\geq 3} \sqrt{{p}_j{q}_j}.			
		\end{align}
	\end{subequations}
	It is straightforward to show that for $t>1/2$ only the first of the above conditions is non-trivial, while for $t<1/2$ only the second one is non-trivial (for $t=1/2$ both of the conditions are trivially fulfilled). We will only focus on the case $t>1/2$, as the other one is proven analogously. First, we lower bound the right hand side of Eq.~\eqref{eq:bracelet_lem} using the bracelet condition:
	\begin{equation}
		\sqrt{\tilde{p}_2\tilde{q}_2}+\sum_{j\geq 3} \sqrt{{p}_j{q}_j}\geq \sqrt{\tilde{p}_2\tilde{q}_2}+\sqrt{p_1q_1}-\sqrt{p_2q_2}.
	\end{equation}
	Then, in order for Eq.~\eqref{eq:bracelet_lem} to hold, it is enough to show that
	\begin{equation}
		\sqrt{p_1q_1}-\sqrt{p_2q_2}\geq \sqrt{\tilde{p}_1\tilde{q}_1}-\sqrt{\tilde{p}_2\tilde{q}_2}.
	\end{equation}
	For that it is sufficient to prove that the function
	\begin{equation}
		f(t):=\sqrt{\tilde{p}_1\tilde{q}_1}-\sqrt{\tilde{p}_2\tilde{q}_2}
	\end{equation}
	is monotonically increasing for \mbox{$t\in[1/2,1]$}. 
	
	From now on, the proof is based on elementary calculus. Calculating the derivative of $f$ over $t$ we get:
	\begin{equation}
		\frac{\mathrm{d}f}{\mathrm{d}t}=\frac{\frac{\mathrm{d}x}{\mathrm{d}t}}{2\sqrt{x}}-\frac{\frac{\mathrm{d}y}{\mathrm{d}t}}{2\sqrt{y}},
	\end{equation}
	where
	\begin{subequations}
	\begin{align}
		\!\!\!x(t):=&t^2 p_1q_1 \!+\!(1\!-\!t)^2 p_2q_2 \!+\!t(1\!-\!t)(p_1q_2+p_2q_1),\\
		\!\!\!y(t):=&t^2 p_2q_2 \!+\!(1\!-\!t)^2 p_1q_1 \!+\!t(1\!-\!t)(p_1q_2+p_2q_1).
	\end{align}
	\end{subequations}
	Note that $x\geq0$ and $y\geq 0$. Moreover, using the two assumptions \mbox{$p_1q_1\geq p_2q_2$} and $t>1/2$, we also have \mbox{$x\geq y$} and \mbox{$\mathrm{d}x/\mathrm{d}t\geq \mathrm{d}y/\mathrm{d}t$}. Thus, to prove that $f(t)$ is monotonically increasing for $t\in[1/2,1]$ we need to show that
	\begin{equation}
		\label{eq:monoton}
		\frac{\mathrm{d}x}{\mathrm{d}t}\sqrt{y}\geq \frac{\mathrm{d}y}{\mathrm{d}t}\sqrt{x}.
	\end{equation}

	There are two cases to be considered. First, if $p_1\geq p_2$ and $q_1\geq q_2$, using our two assumptions it is straightforward to show that $\mathrm{d}x/\mathrm{d}t\geq 0$ and $\mathrm{d}y/\mathrm{d}t\leq 0$. Thus, in that case Eq.~\eqref{eq:monoton} is trivially satisfied. The second case is when $p_1\geq p_2$ and $q_1\leq q_2$ (of course this is equivalent to $p_1\leq p_2$ and $q_1\geq q_2$; while the case $p_1\leq p_2$ and $q_1\leq q_2$ is inconsistent with our assumptions). A direct calculation employing the assumptions shows then that \mbox{$dy/dt\leq 0$}. Hence, if $\mathrm{d}x/\mathrm{d}t\geq 0$, Eq.~\eqref{eq:monoton} is trivially satisfied. Otherwise, both sides of Eq.~\eqref{eq:monoton} are negative, and so it is satisfied when the following is:
	\begin{equation}
		\left(\frac{\mathrm{d}x}{\mathrm{d}t}\right)^2 y\leq \left(\frac{\mathrm{d}y}{\mathrm{d}t}\right)^2 x.
	\end{equation}
	 Again, a direct calculation employing \mbox{$p_1q_1\geq p_2q_2$} and $t>1/2$ shows that the above is fulfilled.
\end{proof}

The second lemma is given as follows.

\begin{lemma}
	\label{lem:rows}
	For $t\in[0,1]$, let $\ket{p}=[t,1-t,0,\dots,0]$ and \mbox{$\ket{q}=[1-t,t,0,\dots,0]$} be vectors of size $d$. Then, for $B\in\L_d$ we have $\{B\ket{p},B\ket{q}\}\in L_d$.
\end{lemma}

\begin{proof}
	Let us denote $\ket{\tilde{p}}=B\ket{p}$ and $\ket{\tilde{q}}=B\ket{q}$. Without loss of generality we may assume that \mbox{${\tilde{p}_1}{\tilde{q}_1} = \max_j 	{\tilde{p}_j}{\tilde{q}_j}$}, and thus to prove the lemma we need to show that 
	\begin{equation}
		\label{eq:aim}
		\sum_{j>1} \sqrt{{\tilde{p}_j}{\tilde{q}_j}}-\sqrt{{\tilde{p}_1}{\tilde{q}_1}} \geq 0.
	\end{equation}
	We begin by using the Cauchy-Schwarz inequality
	\begin{align}
		\sqrt{{\tilde{p}_j}{\tilde{q}_j}} &= \sqrt{\sum_{k,l=1}^2 B_{j k} B_{jl} p_k  q_l} \nonumber\\
		&\geq \sum_{k,l=1}^2 \sqrt{B_{j k} B_{j l} p_k  q_l} X_{kl},\label{eqn:pjpq-ineq}
	\end{align}
	which is satisfied for any $2 \times 2$ matrix $X$, such that $\tra{X^\dagger X}= 1$. 
	
	Now, let the matrix $X^{(1)}$ be chosen in such a way, that the Cauchy-Schwarz inequality in Eq.~\eqref{eqn:pjpq-ineq} is tight for \mbox{$j=1$}, i.e. 
	\begin{equation}
		\sqrt{{\tilde{p}_1}{\tilde{q}_1}} =
		\sum_{k,l=1}^2 \sqrt{B_{1 k} B_{1 l}  p_k  q_l} X_{kl}^{(1)}.
	\end{equation}
	Note, that $X^{(1)}$ has non-negative entries. Writing the above explicitly, we have
	\begin{align}
	\sqrt{\tilde{p}_1 \tilde{q}_1}
	&=
	t \left(
	\sqrt{B_{1 1} B_{1 2} }  X^{(1)}_{12}
	\right) \nonumber\\
	&+
	(1-t) \left(
	\sqrt{B_{1 2} B_{1 1}}  X^{(1)}_{21}
	\right) \nonumber\\
	&+\sqrt{t(1-t)}
	\left(
	B_{1 1}  X^{(1)}_{11}
	+	B_{1 2} X^{(1)}_{22}\right).
	\end{align}
	
	Next, let us consider a matrix $X$ defined by
	\begin{subequations}
	\begin{align}
		X_{11} &= 
		(1-B_{11})
		\frac{\sqrt{|X^{(1)}_{11}|^2+|X^{(1)}_{22}|^2}}
		{\sqrt{(1-B_{11})^2+(1-B_{12})^2}},\\
		X_{22} &= (1-B_{12})
		\frac{\sqrt{|X^{(1)}_{11}|^2+|X^{(1)}_{22}|^2}}
		{\sqrt{(1-B_{11})^2+(1-B_{12})^2}},		\\
		X_{12} &= X^{(1)}_{12}, \\
		X_{21} &= X^{(1)}_{21}.		
	\end{align}
	\end{subequations}
	Note, that $\tra{X X^\dagger} =1$. Since Eq.~\eqref{eqn:pjpq-ineq} holds for any normalised matrix, using matrix $X$ as defined above we get
	\begin{align}
	&\sum_{j>1} \sqrt{\tilde{p_j} \tilde{q_j}} - \sqrt{\tilde{p_1} \tilde{q_1}} \nonumber\\
	&\quad\geq  
	t X^{(1)}_{12} \left(
	\sum_{j>1} \sqrt{B_{j 1} B_{j 2}}  - \sqrt{B_{1 1} B_{1 2}}  
	\right) \nonumber\\
	&\quad+ 
	(1-t) X^{(1)}_{21} \left(
	\sum_{j>1} \sqrt{B_{j 2} B_{j 1}}  - \sqrt{B_{1 2} B_{1 1}}  
	\right) \nonumber\\
	&\quad+ \label{eqn:ch3}
	\sqrt{t(1-t)}
	\left(
	\sum_{j>1}\left( B_{j 1}  X_{11} + B_{j 2} X_{22} \right)\right.\nonumber\\
	&\quad\qquad\qquad\qquad\left.\phantom{\sum_{j>1}}-	B_{1 1}  X^{(1)}_{11} -	B_{1 2} X^{(1)}_{22}
	\right).
	\end{align}
	First, we note that the first two terms on the right hand side of the above inequality are non-negative due to the fact that $B\in\L_d$. To prove that $\{\ket{\tilde{p}},\ket{\tilde{q}}\}\in L_d$ it is thus sufficient to show that the last term is also non-negative.
	
	Using the bistochasticity of $B$ and the definition of $X$ we have
	\begin{align}
	\!\!\!\!& \sum_{j>1}\left( B_{j 1}  X_{11} + B_{j 2} X_{22}\right)
	\nonumber\\
	\!\!\!\!&\quad= 
	(1-B_{1 1})  X_{11} + (1-B_{1 2})  X_{22} \nonumber\\
	\!\!\!\!&\quad=
	{\sqrt{(1\!-\!B_{11})^2+(1\!-\!B_{12})^2}}{\sqrt{|X^{(1)}_{11}|^2+|X^{(1)}_{22}|^2}}.
	\end{align}
	As $B$ is bistochastic we also have 
	\begin{equation}
	(1-B_{11})^2+(1-B_{12})^2 \geq B^2_{11} + B^2_{12}.
	\end{equation}
	Employing this we finally arrive at
	\begin{align}
	& \sum_{j>1} B_{j 1}  X_{11} + B_{j 2} X_{22}
	\nonumber\\
	&\quad\geq 
	{\sqrt{B_{11}^2+B_{12}^2}}{\sqrt{|X^{(1)}_{11}|^2+|X^{(1)}_{22}|^2}}\nonumber\\
	&\quad\geq
	B_{11}X^{(1)}_{11} + B_{12}X^{(1)}_{22},
	\end{align}
	which shows that the last term on the right hand side of Eq.~\eqref{eqn:ch3} is non-negative, and thus implies the desired inequality given in Eq.~\eqref{eq:aim}.
\end{proof}

Finally, the third lemma is given as follows.

\begin{lemma}	
	\label{lem:rows2}
	For $t\in[0,1]$, let $\ket{p}=[t,1-t,0,\dots,0]$ be a vector of size $d$ and $B\in \L_d$ be a bracelet bistochastic matrix. Then, for any $k\geq 3$, we have \mbox{$\{B\ket{p} ,	B\ket{k}\} \in L_d$}, where $\ket{k}$ is $k^{\text{th}}$ unit vector.
\end{lemma}

\begin{proof}	
Without loss of generality we assume that 
\mbox{$\bra{1}B\ket{p} \bra{1}B\ket{k} = \max_j \bra{j}B\ket{p}\bra{j}B\ket{k}$} and thus to prove the lemma we need to show that 
\begin{equation}
\sum_{j>1} \sqrt{\bra{j}B\ket{p}\bra{j}B\ket{k}}\geq
\sqrt{\bra{1}B\ket{p}\bra{1}B\ket{k}}.
\end{equation}
First, we note that for $x\in[0,1]$ and non-negative $a$ and $b$ the following inequality always holds:
\begin{equation}
	\label{eq:ineq}
	\sqrt{a+b}\geq \sqrt{xa}+\sqrt{(1-x)b}.
\end{equation}
Next, using the above term by term and employing the bracelet property of $B$ we get
\begin{align}	
&\!\!\!\sum_{j>1} \sqrt{\bra{j}B\ket{p}\bra{j}B\ket{k}} \nonumber\\
&=
\sum_{j>1} \sqrt{t B_{j1}B_{jk} + (1-t)B_{j2}B_{jk}} \nonumber\\
&\geq 
 \sqrt{x} \sum_{j>1}\sqrt{t B_{j1}B_{jk}} + 
 \sqrt{1-x} \sum_{j>1}\sqrt{(1-t)B_{j2}B_{jk}}\nonumber\\
 & \geq \sqrt{xtB_{11}B_{1k}} + 
 \sqrt{(1-x)(1-t)B_{12}B_{1k}}.
\end{align}
Finally, with the choice 
\begin{equation}
x=\frac{t B_{11} B_{1k}}{t B_{11} B_{1k} + (1-t)B_{12} B_{1k}},
\end{equation}
we get
\begin{align}	
\!\!\!\sum_{j>1} \sqrt{\bra{j}B\ket{p}\bra{j}B\ket{k}}&\geq \sqrt{t B_{11} B_{1k} + (1-t)B_{12} B_{1k}}\nonumber\\
&=\sqrt{\bra{1}B\ket{p}\bra{1}B\ket{k}}.
\end{align}
\end{proof}

We are now ready to present the proof of Theorem~\ref{thm:bracelet}.

\begin{proof}[Proof of Theorem~\ref{thm:bracelet}]
	We will prove that given $B\in\L_d$ and an elementary bistochastic matrix $E$, both products $BE$ and $EB$ belong to $\L_d$. As every factorisable matrix $F\in\F_d$ can be represented as a product of elementary matrices, it will prove that bracelet matrices are closed under multiplication by factorisable matrices. Moreover, since every elementary matrix is bracelet, it will also prove that every factorisable matrix is bracelet.
	
	Without loss of generality we may assume \mbox{$E_{11}=E_{22}=t$}, \mbox{$E_{12}=B_{21}=1-t$}, $E_{kk}=1$ for all $k\geq 3$, and all other entries of $E$ being zero.	Let us focus on $\tilde{B}=EB$, as the argument for $BE$ is analogous, just with columns swapping their roles with rows and vice versa. Let us denote the vectors corresponding to the $k^{\text{th}}$ and $l^{\text{th}}$ column of $B$ by $\ket{p}$ and $\ket{q}$. Then, the $k^{\text{th}}$ and $l^{\text{th}}$ column of $\tilde{B}$ are given by $E\ket{p}$ and $E\ket{q}$. Since $\{\ket{p},\ket{q}\}\in L_d$ (due to $B\in\L_d$), Lemma~\ref{lem:columns} tells us that $\{E\ket{p},E\ket{q}\}\in L_d$. In other words, matrix $\tilde{B}$ satisfies the bracelet condition for columns.
	
	Now, let us show that $\tilde{B}$ also satisfies the bracelet condition for rows. Let $\ket{p}=[t,1-t,0,\dots,0]$ and \mbox{$\ket{q}=[1-t,t,0,\dots,0]$} be vectors of size $d$. Then, the first two rows of $\tilde{B}$ are $B\ket{p}$ and $B\ket{q}$, while the remaining $(d-2)$ rows are the same as for $B$. The first two rows satisfy the bracelet condition due to Lemma~\ref{lem:rows}. Any of the first two rows with the $k^{\text{th}}$ row for $k\geq 3$ satisfy the bracelet condition due to Lemma~\ref{lem:rows2}. Finally, the $k^{\text{th}}$ and $l^{\text{th}}$ row for $k,l\geq 3$ satisfy the bracelet condition because $B\in\L_d$.
\end{proof}


\subsection{Geometric properties}
\label{sec:bracelet_geo}

Although the set of bracelet matrices $\L_d$ is not convex (as seen by the simplest example of a $Q$ matrix from Eq.~\eqref{eq:Q_matrix}), we will now show that both $\L_d$ and its subset $\F_d$ satisfy a weaker condition of being star-shaped with respect to the flat matrix $W_d$.
\begin{theorem}[Star-shape property of $\L_d$ and $\F_d$]
	\label{thm:star}
	The sets of bracelet matrices and factorisable matrices, $\L_d$ and $\F_d$, are star-shaped with respect to the flat van der Waerden matrix $W_d$, i.e., if $B\in \L_d$ (or $B\in \F_d$) then for all $\lambda\in[0,1]$ also
	\begin{equation}
		\label{eq:star}
		B'(\lambda):=(1-\lambda) B+\lambda W_d
	\end{equation}
	belongs to $\L_d$ (or to $\F_d$).
\end{theorem}
\begin{proof}
	First, note that it is enough to prove that for any $d$ and $\lambda \in [0,1]$ matrices $C_d(\lambda)$ from the ray connecting the identity with the van der Waerden matrix,
	\begin{equation}
		C_d(\lambda):=(1-\lambda) \iden_d + \lambda W_d,	
	\end{equation}
	are factorizable. For if this is the case, then given any $B\in \F_d$ we have that
	\begin{equation}
		BC_d(\lambda)=(1-\lambda)B +\lambda W_d=B'(\lambda)
	\end{equation}
	is also factorisable (as it is a product of two factorisable matrices). Also, given any $B\in \L_d$ we have that $B'(\lambda)$ is bracelet due to Theorem~\ref{thm:bracelet} (as it is a product of a bracelet matrix with a factorisable matrix).
	
	The proof will be inductive in the dimension of the system. For $d=2$ every bistochastic matrix belongs to $\F_2$, so that $C_2(\lambda)$ are factorizable for all $\lambda\in[0,1]$. Now, let us assume that for a given dimension $d$ and any $\lambda \in [0,1]$ a matrix $C_d(\lambda)$ belongs to $\F_d$. Note that 
	\begin{equation}
		C_d(\lambda)  = e^{t G_d},
	\end{equation}
	where $t= - \log (1-\lambda)$ and 
	\begin{equation}
		G_d := W_d - \iden_d.
	\end{equation}
	From our assumption, if we embed the matrix $C_d(\lambda)$ in $(d+1) \times (d+1)$ matrices by adding a $1 \times 1$ diagonal block we obtain a matrix from $\F_{d+1}$, i.e. 
	\begin{equation}
		C_d(\lambda) \oplus 1 \in \F_{d+1}.
	\end{equation}
	In the above the diagonal $1 \times 1$ block containing $1$ can be placed on arbitrary position. We also have 
	\begin{equation}
		C_d(\lambda) \oplus 1  = 	e^{t G_d \oplus 0}.
	\end{equation}
	Since direct calculation reveals that 
	\begin{equation}
		G_{d+1} = \frac{d}{(d-1)(d+1)}\sum G_{d}\oplus 0,
	\end{equation}
	where the summation is over all possible placements of the diagonal $1\times 1$ block containing $0$, we get 
	\begin{equation}
		C_{d+1}(\lambda)  =  e^{ t \frac{d}{(d-1)(d+1)}\sum 	G_{d}\oplus 0}.
	\end{equation}
	Now using the above and the Lie product formula,
	\begin{equation}
		e^{A+B}=\lim_{N\rightarrow \infty}\left(e^{A/N}e^{B/N}\right)^N,
	\end{equation}
	we can write $C_{d+1}(\lambda)$ as a limit of product of matrices belonging to the closed set $\F_{d+1}$.	
\end{proof}

Let us also note that if Conjecture~\ref{conj:bracelet} were true, the proof of the star-shape property of $\L_d$ could be significantly simplified. First, one could rewrite Eq.~\eqref{eq:star} as
	\begin{equation}
		B'(\lambda) = [(1-\lambda)\iden_d+\lambda W_d]B.
	\end{equation}
Next, one would note that $(1-\lambda)\iden_d + \lambda_d W$ belongs to $\L_d$, as in the bracelet conditions we always get two copies of the longest segment. Finally, employing Conjecture~\ref{conj:bracelet}, one would obtain that $B'(\lambda)\in\L_d$ as it is a product of two bracelet matrices.

We illustrate Theorem~\ref{thm:star} in Fig.~\ref{fig:star_shaped}, where on two-dimensional cuts through $(d-1)^2$-dimensional Birkhoff's polytopes we indicate in colour those bistochastic matrices that belong to $\L_d$. As can be clearly seen, for any bracelet matrix, also all matrices lying on the line connecting it to the flat matrix $W_d$ (the central point in the plots), belong to $\L_d$.

\begin{figure}[t]
	\begin{center}
		\begin{tikzpicture}
		\node (myfirstpic) at (0,5.5) {\includegraphics[width=0.8\columnwidth]{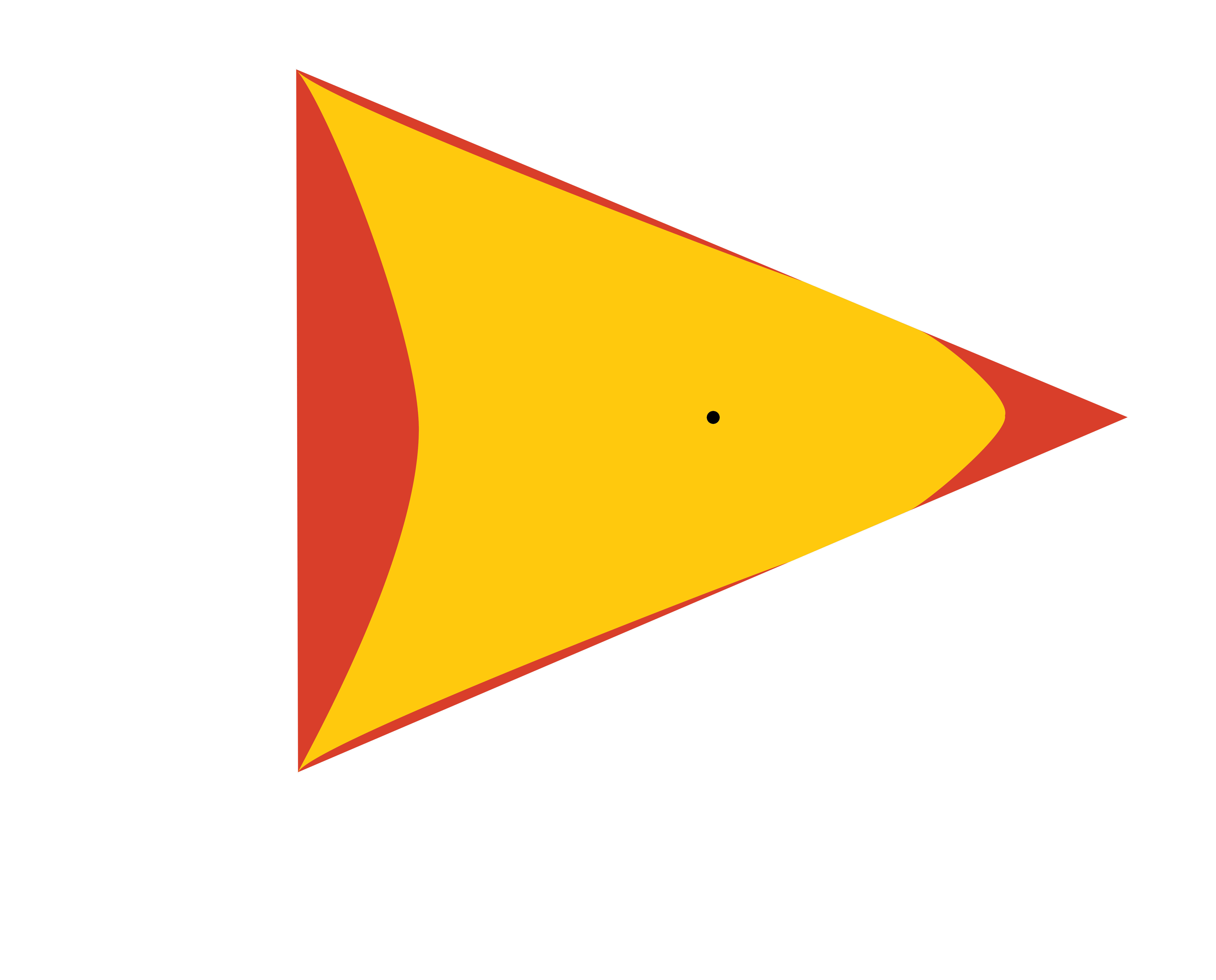}};		
		\node [label={$B=\dfrac{1}{2}\begin{pmatrix}0 &1 &1 &0\\0&0&1&1\\1&0&0&1\\1&1&0&0\end{pmatrix}$}] at (2,6.5){};

		\node [label={$\iden_4$}] at (-2,7.5){};
		\node [label={$W_4$}] at (0.2,5.5){};
		\node [label={$B$}] at (3.3,5.6){};
		\node [label={$\L_4$}] at (-0.5,4.5){};
		\node [label={$\B_4$}] at (-1.4,5.5){};

		\node (myfirstpic) at (0,0) {\includegraphics[width=0.8\columnwidth]{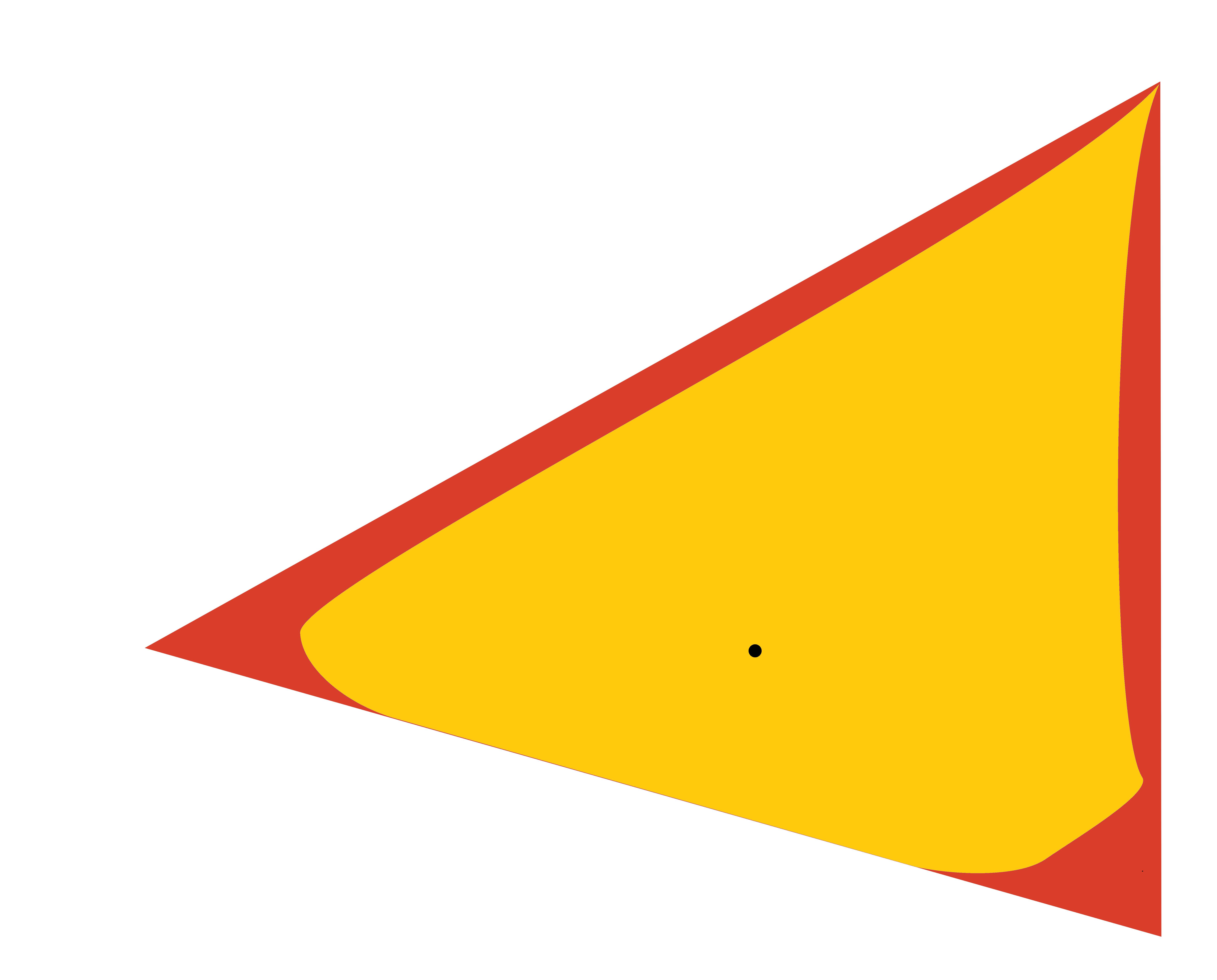}};		
		\node[label={$B=\dfrac{1}{3}\begin{pmatrix}1&0 &1 & 0 &1\\1&1&0&1&0\\0&1&1&0&1\\1&0&1&1&0\\0&1&0&1&1\end{pmatrix}$}] at (-1.5,0.5){};
		\node [label={$\iden_5$}] at (3.5,2.1){};
		\node [label={$\B_5$}] at (2.9,-2.55){};
		\node [label={$\L_5$}] at (2,0.3){};
		\node [label={$W_5$}] at (0.5,-1.5){};
		\node [label={$B$}] at (-2.9,-1.3){};
		\end{tikzpicture}
	\end{center}
	\caption{\textbf{Star-shaped property of bracelet matrices.} Two-dimensional cross-sections through the Birkhoff polytope (top: $d=4$, bottom: $d=5$) containing the flat matrix $W_d$, identity matrix $\iden_d$ and a bistochastic matrix $B$. Matrices belonging to the set of bracelet matrices $\mathcal{L}_d$ are indicated in yellow, while the remaining bistochastic matrices are presented in red. The graph was plotted using barycentric coordinates with the flat matrix $W_d$ lying in the centre $(0,0)$, illustrating that although the set $\mathcal{L}_d$ is not convex, it is star-shaped with respect to $W_d$. \label{fig:star_shaped}}
\end{figure}


\subsection{Tensor structure}

Finally, we also present the result concerning the behaviour of bracelet matrices under the tensor product.
\begin{theorem}[Tensor structure of $\L_d$]
	If $B_1\in \L_{d_1}$ and $B_2\in\L_{d_2}$ then \mbox{$B=B_1 \otimes B_2$}{} belongs to~$\L_{d_1d_2}$.
\end{theorem}
\begin{proof}
	Without loss of generality we may focus on the bracelet conditions for the first two columns of $B$. We have
	\begin{subequations}
	\begin{align}
		\ket{p} &= B \ket{0} = B_1 \ket{0} \otimes B_2 \ket{0} =: \ket{p^{(1)}} \otimes \ket{p^{(2)}}, \\
		\ket{q} &= B \ket{1} = B_1 \ket{0} \otimes B_2 \ket{1} =: \ket{p^{(1)}} \otimes \ket{q^{(2)}}.
	\end{align}
	\end{subequations}
	Since both $B_1$ and $B_2$ are bracelet, we can use the bracelet condition for vectors $p^{(2)}$ and $q^{(2)}$ to obtain the following
	\begin{align}
	&\sqrt{\matrixel{0}{B}{0} \matrixel{0}{B}{1}}  \non
	&\qquad=\sqrt{p^{(1)}_0 p^{(1)}_0}\sqrt{p^{(2)}_0  q^{(2)}_0}\non
	&\qquad\leq \sum_{(j,k) \neq (0,0)}	\sqrt{p^{(1)}_j p^{(1)}_j}	\sqrt{p^{(2)}_k q^{(2)}_k}\non
	&\qquad=
	\sum_{(j,k) \neq (0,0)}
	\sqrt{\bra{j} B_1 \ket{0} \bra{k} B_2 \ket{0}}\sqrt{\bra{j} B_1 \ket{0} \bra{k} B_2 \ket{1}} \non
	&\qquad=
	\sum_{(j,k) \neq (0,0)}
	\sqrt{\bra{jk}B\ket{00} \bra{jk} B\ket{01}} \non
	&\qquad=
	\sum_{k\neq 0}
	\sqrt{\bra{k}B\ket{0} \bra{k} B \ket{1}}.
	\end{align}

\end{proof}


\section{Structure of the set of unistochastic matrices}
\label{sec:uni}

\subsection{Low-dimensional cases}

As already explained in Sec.~\ref{sec:preliminaries}, unistochastic matrices of size $d$ form a subset of the Birkhoff polytope $\B_d$. Unfortunately, it is unclear how to extend some low-dimensional properties of the Birkhoff polytope to higher dimensions~\cite{jaekel2011polytopes,pak2000Birkhoff}, and proving general statements concerning the unistochastic set is even harder. Thus, we shall first address several low dimensions individually to build up some understanding of $\U_d$.

For $d=2$ it is well known that the sets of bistochastic and unistochastic matrices coincide. More precisely, the general $2\times 2$ bistochastic matrix and the corresponding unitary matrix (chosen to be circulant for later convenience) are given by
\begin{equation}
	\label{eq:2x2_uni}
	B=\begin{pmatrix}
			a&1-a\\
			1-a& a
	\end{pmatrix},\quad 
	U=\begin{pmatrix}
			i\sqrt{a}&\sqrt{1-a}\\
			\sqrt{1-a}&i\sqrt{a}
	\end{pmatrix}.
\end{equation}
As a result, the set of $2\times 2$ unistochastic matrices $\U_2$ shares all the geometric and algebraic properties with $\B_2$, i.e., it is convex and closed under matrix multiplication. 

For $d=3$ it has been established that the bracelet inequalities are not only necessary, but also sufficient~\cite{bengtsson2005birkhoff}. Therefore, $\U_3=\L_3$, and so employing our results, Theorem~\ref{thm:bracelet} and Theorem~\ref{thm:star}, we arrive at the following.
\begin{corollary}[Structure of $\U_3$]
	The set of $3\times 3$ unistochastic matrices is closed under multiplication by factorisable matrices and is star-shaped with respect to the flat matrix $W_3$.
\end{corollary}
We note that the star-shaped property of $\U_3$ was previously observed and is actually conjectured to hold for all $d$~\cite{bengtsson2005birkhoff}. Moreover, if Conjecture~\ref{conj:bracelet} were true (as the numerical evidence strongly suggests for $d=3$), mathematically it would motivate the introduction of a new product operation $\times$ between unitary matrices of size 3 defined via the product of the corresponding unistochastic matrices. More precisely, for unitary matrices $U$ and $V$ with the corresponding unistochastic matrices $B_U$ and $B_V$, the product $U \times V$ would be given by the unitary corresponding to the unistochastic matrix $B_U B_V$. Note that such a product would not be unique as it is defined up to phase transformations (left and right multiplication by diagonal unitary matrices), but it would become unique if one focused on unitary matrices in the dephased form, i.e., with the first row and column being real and positive. On the other hand, physically we would see that every qutrit quantum channel defined by consecutive unitary evolutions $U_i$ followed by measurements in the computational basis could be realised by a single unitary evolution $U$ followed by a single measurement in the same basis.

For $d=4$ bracelet conditions are only necessary and not sufficient for unistochasticity, i.e., there are matrices in $\L_4$ that do not belong to $\U_4$. The simplest way to see this is to recall that the flat matrix $W_4$ lies on the boundary of unistochastic set, and so in every neighbourhood around the centre of the Birkhoff polytope there are non-unistochastic matrices~\cite{bengtsson2005birkhoff}. On the other hand, it is easy to see that there exists a ball with finite radius around $W_4$ with all matrices within the ball fulfilling the bracelet conditions. Moreover, a statement analogous to Conjecture~\ref{conj:bracelet} for unistochastic matrices will not hold, because the set $\U_4$ does not form a monoid. A simple counterexample is given by the following unistochastic matrix
\begin{equation}
B=\frac{1}{100}
\begin{pmatrix}
24 & 16 & 35 & 25 \\
38 & 21 & 12 &29 \\
23 & 24 & 14 & 39 \\
15 & 39 & 39 & 07
\end{pmatrix},
\end{equation}
the square of which, $B^2$, is not unistochastic. Here, to verify the unistochasticity of a $4\times 4$ bistochastic matrix, we used the algorithm of Ref.~\cite{rajchel2018unistochastic}. What is more, numerical results obtained there suggest that the set of matrices that fulfil the bracelet conditions but are not unistochastic is not of measure zero. We illustrate the above considerations in Fig.~\ref{fig:4x4}, where one can clearly see the difference between $\L_4$ and $\U_4$.

\begin{figure}[t]
	\begin{center}
		\begin{tikzpicture}
			\node (myfirstpic) at (0,0) {\includegraphics[width=0.8\columnwidth]{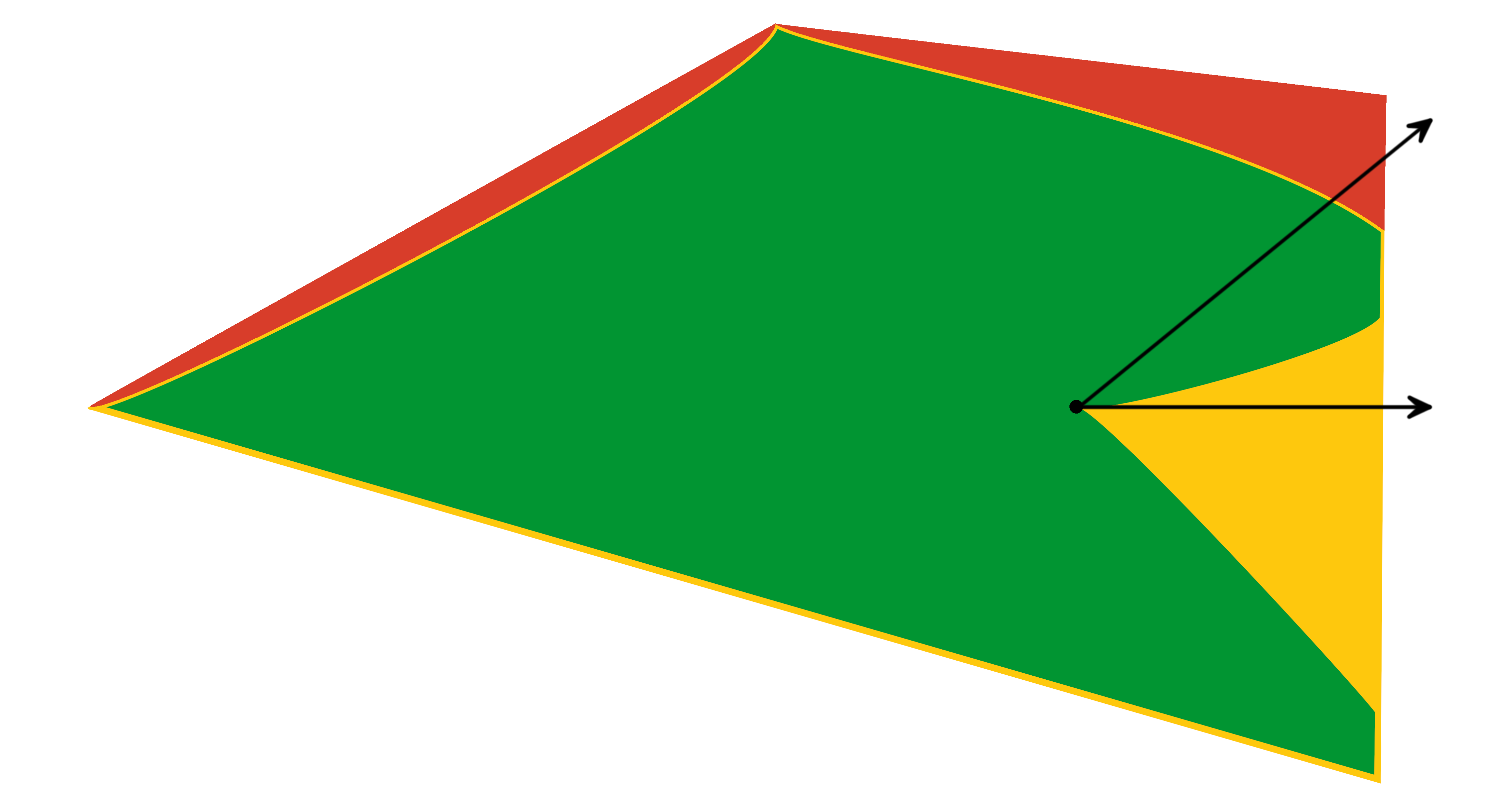}};		
			\node [label={$v_1=\begin{pmatrix}9 &-3 &-3 &-3\\-3&1&1&1\\-3&1&1&1\\-3&1&1&1\end{pmatrix}\!,$}] at (-2.3,-4){};
			\node [label={$v_2=\begin{pmatrix}7 &-1 &-1 &-5\\-1&-1&-1&3\\-1&-1&-1&3\\-5&3&3&-1\end{pmatrix}$}] at (1.5,-4){};	
			
			\node at (0.14\columnwidth,-0.0\columnwidth) {\color{black}$W_4$};
			\node at (0.29\columnwidth,-0.07\columnwidth) {\color{black}$\L_4$};
			\node at (-0.05\columnwidth,0.01\columnwidth) {\color{black}$\U_4$};
			\node at (0.42\columnwidth,-0.02\columnwidth) {\color{black}$v_1$};	
			\node at (0.42\columnwidth,0.16\columnwidth) {\color{black}$v_2$};		
			\node at (0.295\columnwidth,0.135\columnwidth) {\color{black}$\B_4$};					
		\end{tikzpicture}
	\end{center}
	\caption{\textbf{$\U_4$ is strictly smaller than $\L_4$.} Two-dimensional cross-section through the Birkhoff polytope with $d=4$ containing the flat matrix $W_4$. The graph was plotted using barycentric coordinates with the flat matrix $W_4$ lying in the centre $(0,0)$, illustrating that $W_4$ lies at the boundary of the unistochastic set, as no matrix on the $v_1$ line is unistochastic.\label{fig:4x4}}
\end{figure}

However, for $d=4$ the situation changes drastically when we restrict our considerations to circulant unistochastic matrices, which is captured by the following theorem.
\begin{theorem}[Structure of $\C_4\cap\U_4$]
	\label{thm:4x4}
	The set of circulant unistochastic matrices, $\C_4\cap\U_4$, coincides with the set of circulant bracelet matrices $\C_4\cap\L_4$. Thus, $\C_4\cap\U_4$ is closed under multiplication by circulant factorisable matrices and is star-shaped with respect to the flat matrix~$W_4$.
\end{theorem}
\begin{proof}
First, we will show that a general $4\times 4$ circulant matrix $B$ that satisfies the bracelet condition, \mbox{$B\in \C_4\cap\L_4$}, is also unistochastic. Let us parametrise $B$ as follows:
\begin{equation}
    B = \begin{pmatrix}
    a & b & c & d \\
    d & a & b & c\\
    c & d & a & b \\
    b & c & d  & a
    \end{pmatrix},
\end{equation}
with $a+b+c+d = 1$. We will assume $ac$ or $bd$ is greater than 0, since if both of them are equal 0 then $B$ is not a bracelet matrix. Furthermore, without loss of generality, let us assume that $ac\leq bd$ (for if this is not the case, we can change the position of the rows and at the end we can recover a unitary matrix corresponding to the initial one by changing the rows back).

Now, let us create a circulant matrix by taking an element-wise square root of $B$ and providing phases $\alpha$, $\beta$, $\gamma \in [0,2\pi)$:
\begin{equation}
    M = \begin{pmatrix}
    \sqrt{a} & e^{i\alpha}\sqrt{b} & e^{i\beta}\sqrt{c} & e^{i\gamma}\sqrt{d} \\
    e^{i\gamma}\sqrt{d} & \sqrt{a} & e^{i\alpha}\sqrt{b} & e^{i\beta}\sqrt{c}\\
    e^{i\beta}\sqrt{c} & e^{i\gamma}\sqrt{d} & \sqrt{a} & e^{i\alpha}\sqrt{b} \\
    e^{i\alpha}\sqrt{b} & e^{i\beta}\sqrt{c} & e^{i\gamma}\sqrt{d}  & \sqrt{a}
    \end{pmatrix}.
\end{equation}
Our aim is to to show that there exists a choice of angles such that $MM^\dagger =\mathbbm{1}$, and so that $B$ is a circulant unistochastic matrix (in fact, since $M$ is circulant, we will prove that $B$ is a doubly circulant unistochastic matrix).

Due to symmetries of a circulant matrix, $MM^\dagger$ is entirely defined by a complex number $m$ and a real number $n$
\begin{equation}
    MM^\dagger = \begin{pmatrix}
    1 & m^* & n & m \\
    m & 1 & m^* & n\\
    n & m & 1 & m^* \\
    m^* & n & m  & 1
    \end{pmatrix},
\end{equation}
where 
\begin{equation}\label{eq:m_proof}
    m = e^{i\gamma}\sqrt{ad} + e^{-i\alpha}\sqrt{ab} + e^{i(\alpha-\beta)}\sqrt{bc} + e^{i(\beta-\gamma)}\sqrt{cd}
\end{equation}
and 
\begin{align}\label{eq:n_proof}
    n& = e^{i\beta}\sqrt{ac} + e^{i(\gamma-\alpha)}\sqrt{bd} + e^{-i\beta}\sqrt{ac} + e^{i(\alpha - \gamma)}\sqrt{bd}  \non
    & = 2 \sqrt{ac}\cos{\beta} + 2 \sqrt{bd}\cos{(\alpha - \gamma)}.
\end{align}

If $m, n$ are to be zero, using Eq.~(\ref{eq:n_proof}) we obtain that $\alpha = \gamma + \arccos{(-\sqrt{\frac{ac}{bd}}\cos{\beta})} = \gamma + f(\beta)$ with new notation $f(\beta)$.
For our purposes $\arccos: [-1,1] \to [0,\pi]$. We can insert this value into Eq.~(\ref{eq:m_proof}) and obtain
\begin{align}
    m =& e^{i\gamma}\bigg(\sqrt{ad}+e^{i(f(\beta)-\beta)}\sqrt{bc}\bigg) \non
   & +e^{-i\gamma}\bigg(e^{ - if(\beta))}\sqrt{ab} + e^{i\beta}\sqrt{cd}\bigg).
\end{align}
The value of $m$ will also be zero if
\begin{equation}\label{eq:gamma_proof}
     e^{2i\gamma} = - \frac{e^{ - if(\beta))}\sqrt{ab} + e^{i\beta}\sqrt{cd}}{e^{i(f(\beta)-\beta)}\sqrt{bc}+\sqrt{ad}}.
\end{equation}

On the whole, if Eq.~(\ref{eq:gamma_proof}) has solutions (the right-hand side has modulus 1), then there exist angles for which $M$ is a unitary matrix. In order to show that this is always the case, let us consider properties of a function $F: [\pi/2,3\pi/2] \to [0,2\pi)^2$ defined by 
\begin{equation}
    F(x) = \big(x + \arccos{(\eta \cos{x})},x - \arccos{(\eta \cos{x})}\big),
\end{equation}
with $\eta = -\sqrt{\frac{ac}{bd}} \in [-1,0]$ since we have chosen $ac\leq bd$.
The important property of this function is that for these values of $\eta$ it is continuous with $F(\pi/2) = (\pi,0)$ and $F(3\pi/2) = (0,\pi)$.

Now we define a function $G: [0,2\pi)^2 \to \mathbb{R}$ with positive parameters $a, b, c, d \leq 1$, defined as
\begin{equation}
    G(x,y) = |\sqrt{ab} + e^{ix}\sqrt{cd} | - |\sqrt{bc}+e^{iy}\sqrt{ad}|.
\end{equation}
Since $G$ is continuous, $G \circ F$ is continuous with the values of opposite signs at the boundaries of the domain:
\begin{subequations}
\begin{align}
    G(F(\pi/2)) =&  |\sqrt{ab} - \sqrt{cd} | - |\sqrt{bc}+\sqrt{ad}| \non
    =& \max{(\sqrt{ab},\sqrt{cd})}-\min{(\sqrt{ab},\sqrt{cd})}\non
    &- \sqrt{bc}- \sqrt{ad} \leq 0,
\end{align}
\begin{align}
    G(F(3\pi/2)) =&  |\sqrt{ab} + \sqrt{cd} | - |\sqrt{bc} - \sqrt{ad}| \non
    =& \min{(\sqrt{bc},\sqrt{ad})} - \max{(\sqrt{bc},\sqrt{ad})}\non
    &+\sqrt{ab} + \sqrt{cd}  \geq 0,
\end{align}
\end{subequations}
where inequalities hold if $B$ is a bracelet matrix.

Using Bolzano's theorem, we now deduce that there exists some value of $\beta$ for which $G(F(\beta)) =0$.
The same $\beta$ yields a fraction of modulus 1 in Eq.~(\ref{eq:gamma_proof}), thus providing us with $\gamma$ and $\alpha = \gamma+ f(\beta)$ for which $M$ is a unitary matrix.

Finally, knowing that $\C_4\cap\U_4$ coincides with $\C_4\cap\L_4$, let us prove that $\C_4\cap\U_4$ is closed under multiplication by circulant factorisable matrices $\C_4\cap\F_4$ and is star-shaped with respect to the flat matrix~$W_4$. To prove the first statement, simply note that a product of two circulant matrices is still circulant and, by Theorem~\ref{thm:bracelet}, a product of a bracelet matrix with a factorisable matrix is bracelet. To prove the second one, note that a convex sum of a flat matrix with a circulant matrix is circulant and, by Theorem~\ref{thm:star}, the set of bracelet matrices is star-shaped.
\end{proof}

\begin{figure}[t]
	\begin{tikzpicture}
		\node (myfirstpic) at (0,0) {\includegraphics[width=0.65\columnwidth]{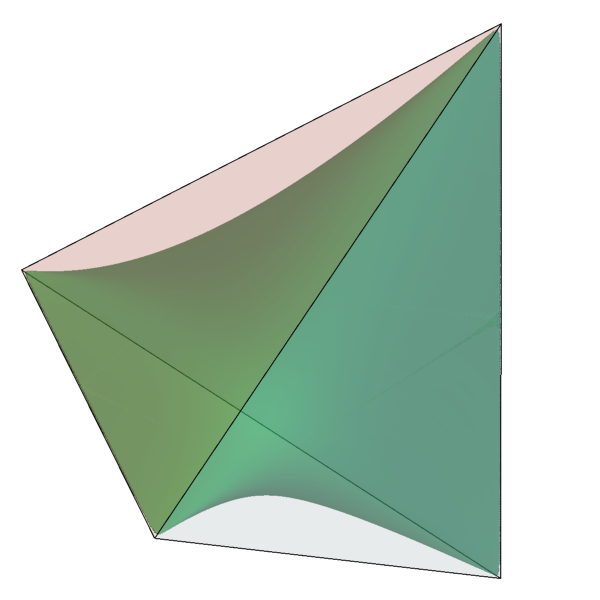}};		

		\node at (-0.23\columnwidth,-0.3\columnwidth) {\small\color{black}$\iden_4$};
		\node at (0.28\columnwidth,-0.32\columnwidth) {\small\color{black}$\Pi_4$};
		\node at (-0.32\columnwidth,0.1\columnwidth) {\small\color{black}$\Pi_4^2$};
		\node at (0.28\columnwidth,0.32\columnwidth) {\small\color{black}$\Pi_4^3$};		

		\node [label={$\Pi_4=\begin{pmatrix}0 &1 &0 &0\\0&0&1&0\\0&0&0&1\\1&0&0&0\end{pmatrix}$}] at (-0.2\columnwidth,0.2\columnwidth){};
	\end{tikzpicture}
	\caption{\label{fig:tetra} \textbf{Circulant unistochastic matrices of size $d=4$.} The set of $4\times 4$ circulant matrices can be represented by a tetrahedron with vertices corresponding to four circulant permutation matrices $\Pi_4^n$, and the points inside to convex combinations of these extremal matrices. The subset of circulant matrices that additionally satisfy the bracelet condition (so, according to Theorem~\ref{thm:4x4}, the set of circulant unistochastic matrices $\C_4\cap\U_4$) is presented as the non-convex, green shape inside the tetrahedron. The edges connecting $\iden_4$ with $\Pi_4^2$ and $\Pi_4$ with $\Pi_4^3$ correspond to complementary permutation matrices~\cite{AYC91,rajchel2018unistochastic} and are unistochastic. The only remaining points at which the set of circulant unistochastic matrices touches the surface of the tetrahedron are given by the heights connecting the midpoints of these edges with the other two permutation matrices.}
\end{figure}

The geometric structure of $\C_4\cap\U_4$ captured by Theorem~\ref{thm:4x4} is illustrated in Fig.~\ref{fig:tetra}. Notice that the neighbourhood of $W_4$ when restricted to circulant matrices is unistochastic, which is not the case if we consider the whole Birkhoff polytope $\B_4$~\cite{bengtsson2005birkhoff}. Furthermore, as with $3\times 3$ unitary matrices, if Conjecture~\ref{conj:bracelet} were true we could introduce a new product operation $\times_c$ between circulant unitary matrices of size 4 defined via the product of the corresponding unistochastic matrices. 


\subsection{Doubly circulant unistochastic matrices}

We will now take a closer look at the set of circulant unistochastic matrices and doubly circulant unistochastic matrices. Let us start with the following result.
\begin{lemma}[Circulant and doubly circulant sets]\label{lemma:circulant}
	In dimensions $d\leq 4$ the sets of circulant unistochastic matrices, $\C_d\cap\U_d$, coincide with the sets of doubly circulant unistochastic matrices $\C_d^2$.
\end{lemma}
\begin{proof}
	For $d=2$, the circulant unitary matrix corresponding to a general circulant unistochastic matrix is given by Eq.~\eqref{eq:2x2_uni}. For $d=3$, the circulant bistochastic matrix $B$ defined by a column vector \mbox{$[a,b,c]$} with \mbox{$c=1-a-b$} is unistochastic if and only if there exists a unitary matrix $U$ in the dephased form
	\begin{equation}
	U=\begin{pmatrix}
		\sqrt{a} & \sqrt{b} &\sqrt{c}\\
		\sqrt{c} & e^{ix}\sqrt{a}& e^{i\tilde{x}}\sqrt{b}\\
		\sqrt{b} & e^{iy}\sqrt{c} & e^{i\tilde{y}}\sqrt{a}
	\end{pmatrix}.
	\end{equation}
	A straightforward calculation then shows that the following matrix,
	\begin{equation}
		V=\begin{pmatrix}
			e^{i\frac{2x}{3}}\sqrt{a} & e^{i\frac{x-y}{3}}\sqrt{b} & e^{i\frac{y}{3}}\sqrt{c}\\
			e^{i\frac{y}{3}}\sqrt{c} & e^{i\frac{2x}{3}}\sqrt{a}& e^{i\frac{x-y}{3}}\sqrt{b}\\
			e^{i\frac{x-y}{3}}\sqrt{b} & e^{i\frac{y}{3}}\sqrt{c} & e^{i\frac{2x}{3}}\sqrt{a}
		\end{pmatrix},
	\end{equation}
	is circulant and unitary, and  corresponds to the same unistochastic matrix $B$. Finally, for $d=4$, the proof of Theorem~\ref{thm:4x4} contains the construction of $4\times 4$ unitary circulant matrix for every circulant unistochastic matrix of size 4.
\end{proof}

Observation from the above lemma suggests the following conjecture.
\begin{conjecture}
	The sets $\C_d\cap\U_d$ and $\C_d^2$ coincide for all $d$, i.e., for every circulant unistochastic matrix it is possible to find a corresponding circulant unitary matrix.
\end{conjecture}

We now proceed to characterising the spectrum of doubly circulant unistochastic matrices. Note that this way we also characterise the spectrum of circulant unistochastic matrices with $d\leq 4$ (and for all $d$ if the above conjecture holds). The support of the spectra of unistochastic matrices of order $d$ was studied in \cite{ZKSS03}, where it was shown that it includes all $k$--hypocycloids with $k=2,\dots d$, and some regions between the $k$-th roots of identity. This structure is analogous to the spectra of stochastic matrices, restricted by smooth interpolating bounds of Karpelevich \cite{Ka51}, later simplified by Djokovi{\v c} \cite{Dj90}. For doubly circulant unistochastic matrices we have the following result that we illustrate in Fig.~\ref{fig:hypo}.

\begin{figure}[t]
	\begin{center}
		\includegraphics[width=0.75\columnwidth]{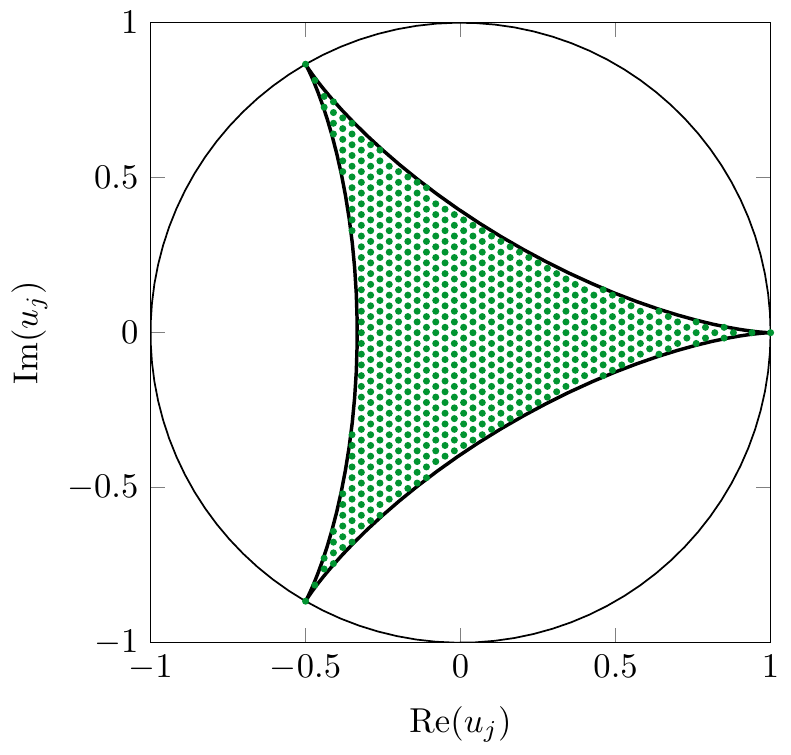}
		\includegraphics[width=0.75\columnwidth]{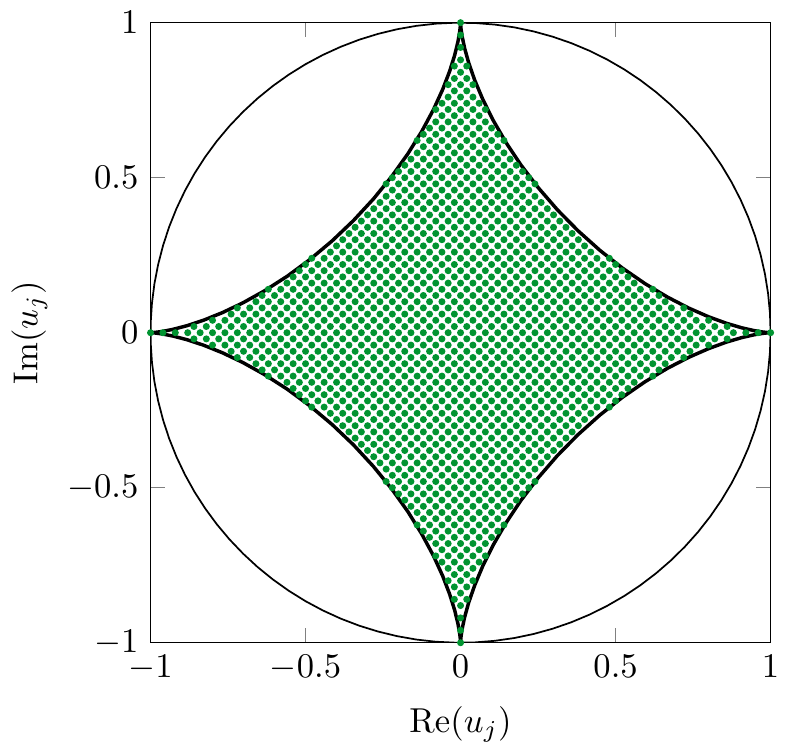}
	\end{center}
	\caption{\textbf{Eigenvalues of doubly circulant unistochastic matrices.} Eigenvalues of $3\times 3$ (top) and $4\times 4$ (bottom) doubly circulant unistochastic matrices in the complex plane (green dots) together with
		deltoid and astroid, i.e. 3- and 4-hypocycloids, respectively. 
		The plotted set of matrices was obtained by discretising 3- and 4-dimensional probability vectors describing a given circulant matrix using a grid of size 0.02, and verifying whether the resulting matrix is a bracelet matrix.\label{fig:hypo}}
\end{figure}

\begin{theorem}[Eigenvalues of $\C_d^2$] 
	The eigenvalues of a $d\times d$ doubly circulant unistochastic matrix $B$ lie inside a unit $d$-hypocycloid $H_d$ on a complex plane, i.e., inside the star-shaped region with the boundary parametrised by
	\begin{subequations}
		\begin{align}
		x(\theta)&=\frac{d-1}{d}\cos\theta+\frac{1}{d}\cos((d-1)\theta),\label{eq:hypo1}\\
		y(\theta)&=\frac{d-1}{d}\sin\theta-\frac{1}{d}\sin((d-1)\theta).\label{eq:hypo2}
		\end{align} 
	\end{subequations}
\end{theorem}
\begin{proof}
	Consider $B$ being unistochastic and doubly circulant, i.e., that there exists a circulant unitary $U$ such that \mbox{$B_{ij}=|\!\matrixel{i}{U}{j}\!|^2$}. Introducing a shift matrix $\Delta$,
	\begin{equation}
		\Delta:=\sum_{j=0}^{d-1} \ketbra{j\oplus 1}{j},
	\end{equation}
	with $\oplus$ denoting addition modulo $d$, we have
	\begin{align}
		U=\sum_{j=0}^{d-1} a_j \Delta^j,\quad B=\sum_{j=0}^{d-1} |a_j|^2 \Delta^j.
	\end{align}
	Moreover, as $U$ and $B$ are circulant, they are diagonalised by a Fourier matrix $F$, so that
	\begin{equation}
		U=F \Lambda_U F^\dagger,\quad	B=F \Lambda_B F^\dagger,\quad
	\end{equation}
	where $\Lambda_U$ and $\Lambda_B$ denote the diagonal matrices of the eigenvalues of $U$ and $B$, respectively. 
	
	Now, the eigenvalues of $\Delta$ are given by the $d$-th roots of unity, and so the eigenvalues of $U$ and $B$ are given by:
	\begin{subequations}
	\begin{align}
		u_j&:=[\Lambda_U]_{jj}=\sum_{k=0}^{d-1} a_k \exp\left(i\frac{2 \pi kj}{d}\right),\label{eq:uj}\\
		b_j&:=[\Lambda_B]_{jj}=\sum_{k=0}^{d-1} |a_k|^2 \exp\left(i\frac{2 \pi kj}{d}\right).\label{eq:bj}
	\end{align}
	\end{subequations}
	By applying the inverse Fourier transform to Eq.~\eqref{eq:uj} we get
	\begin{equation}
		a_k=\frac{1}{d}\sum_{j=0}^{d-1} u_j \exp\left(-i\frac{2 \pi kj}{d}\right),
	\end{equation}
	which can be substituted to Eq.~\eqref{eq:bj} to yield
	\begin{align}
		b_k&=\frac{1}{d} \sum_{j=0}^{d-1} u_j u_{j-k}^*.
	\end{align}
	Next, we rewrite the above as
	\begin{align}
		b_k&=\frac{1}{d} \tra{\Lambda_U \Delta^k \Lambda_U^\dagger (\Delta^\dagger)^k}
	\end{align}
	and observe that the determinant of the matrix under the trace is equal to 1. Hence, $b_k$ are given by the average of the eigenvalues of a special unitary matrix and so by the result of Ref.~\cite{kaiser2006mean} they lie on the complex plane within unit $d$-hypocycloid. 
\end{proof}

Finally, we observe the following. Given two doubly circulant unistochastic matrices, $B_1$ and $B_2$, the eigenvalues of their product, $B=B_1 B_2$, are given by the product of their eigenvalues $\Lambda_B=\Lambda_{B_1}\Lambda_{B_2}$. That means that each resulting eigenvalue corresponds to a product of two points within a $d$-hypocycloid $H_d$. However, as we prove below, such a product also lies inside a $d$-hypocycloid. Therefore, even though it is not clear for all $d$ whether doubly circulant unistochastic matrices form a monoid, they nevertheless preserve the structure of their spectrum (constrained to a hypocycloid) under matrix multiplication.

\begin{lemma}
\label{lem18}
	Given any two points lying inside a $d$-hypocycloid, $z_1,z_2\in H_d$, we have $z_1z_2\in H_d$.
\end{lemma}
\begin{proof}
	First, by denoting the boundary of $H_d$ in polar coordinates by $r(\phi)$, we can express the region $H_d$ as follows
	\begin{equation}
		H_d=\{z\in\mathbb{C} : |z|\leq r(\arg{z}) \}.
	\end{equation} 
	Next, let us assume that the function $-\log r(\phi)$ is sub-additive, meaning
	\begin{equation}
		r(\phi_1)r(\phi_2)\leq r(\phi_1+\phi_2).
	\end{equation}
	For $z_1,z_2\in H_d$, we then have
	\begin{align}
		|z_1z_2| \leq r(\arg z_1) r(\arg z_2)&\leq r(\arg z_1+\arg z_2)\nonumber\\
		&\leq r(\arg(z_1z_2)),
	\end{align}
	and so $z_1z_2\in H_d$. 
	
	Therefore, it remains to show that $-\log r(\phi)$ is sub-additive. Since the studied function is non-negative and defined on a non-negative domain, it is sufficient to prove that it is concave, i.e., that
	\begin{equation}
		\frac{\mathrm{d}^2 (-\log r)}{\mathrm{d}\phi^2} \leq 0.		
	\end{equation}
	To achieve this we will use the chain rule twice. First, we have
	\begin{equation}
	\label{eq:chain_rule}
	\frac{\mathrm{d} (-\log r)}{\mathrm{d}\phi}=\frac{\mathrm{d}(-\log r)}{\mathrm{d}\theta}\frac{\mathrm{d}\theta}{\mathrm{d}\phi}.
	\end{equation}	
	Using Eqs.~\eqref{eq:hypo1}-\eqref{eq:hypo2}, we find that $r$ as a function of parameter $\theta$ is given by
	\begin{align}
		r(\theta) =&\sqrt{x^2(\theta)+y^2(\theta)}\nonumber\\
		=&\sqrt{\frac{2+(d-2)d+2(d-1)\cos(d\theta)}{d^2}}. 
	\end{align}
	Employing the same equations, we find that
	\begin{align}
		\frac{\mathrm{d}\phi}{\mathrm{d}\theta}=&\frac{\mathrm{d}\arctan(y(\theta)/x(\theta))}{\mathrm{d}\theta}\nonumber\\
		=&\frac{2(d-2)(d-1)\sin^2\left(\frac{d\theta}{2}\right)}{2+(d-2)d+2(d-1)\cos\left(d\theta\right)}.
	\end{align}
	Substituting the above to Eq.~\eqref{eq:chain_rule}, taking the second derivative over $\phi$ and substituting again, one ends up with
	\begin{equation}
		\frac{\mathrm{d}^2 (-\log r)}{\mathrm{d}\phi^2} =- \frac{d^2(2\!+\!(d\!-\!2)d\!+\!2(d\!-\!1)\cos(d\theta))}{4(d\!-\!2)^2(d\!-\!1)\sin^4\left(\frac{d\theta}{2}\right)}.
	\end{equation}
	By direct inspection one can see that the above is always negative for $d\geq 3$. For $d=2$ the 2-hypocycloid is segment connecting -1 and 1, so the product of any two points belonging to $H_2$ also belongs to it; while for $d=1$ it is just a single point 1, trivially satisfying the condition. 
\end{proof}
 
As a final remark, let us note that analysing sets in the complex plane it is convenient to use the notion of their \emph{Minkowski product}~\cite{FMR01}:
\begin{equation}
Z_1 \boxtimes  Z_2 =  
\left\{z: z=z_1 z_2, \ z_1 \in Z_1, \ z_2 \in Z_2 \right\} .
\label{Mprod} 
\end{equation}
Lemma~\ref{lem18} then shows that the set $H_d$ bounded by the $d$-hypocycloid inscribed in the unit disk
is idempotent with respect to the Minkowski product: $H_d \boxtimes  H_d = H_d$. This geometric result implies that the spectrum of a product of two doubly circulant unistochastic matrices belongs to the support of spectra of doubly circulant unistochastic matrices. Thus, this observation is consistent with a conjecture that doubly circulant unistochastic matrices form a monoid. In fact, for $d=3$ it is straightforward to prove that $\C_3\cap\U_3$ forms a monoid. To show this, one simply notes that the bracelet condition for $3 \times 3$ circulant matrices is equivalent to the condition on eigenvalues to lie inside a 3-hypocycloid~\cite{bengtsson2005birkhoff}.


\section{Conclusions and outlook}
\label{sec:conclusions}

In this paper we have addressed the problem of characterising the set $\U_d$ of unistochastic matrices of size $d$ by introducing and analysing the properties of its superset $\L_d$ of bracelet matrices. We have also investigated the special case of circulant unistochastic matrices, $\C_d\cap \U_d$, and introduced the notion of doubly circulant unistochastic matrices $\C_d^2$ (such elements of $\C_d\cap \U_d$ that can arise from circulant unitary matrices). We proved that $\L_d$ is closed under matrix multiplication by factorisable bistochastic matrices $\F_d$, and is star-shaped with respect to the flat van der Waerden matrix $W_d$. Moreover, we have shown that the spectra of doubly circulant unistochastic matrices of size $d$ lie inside a $d$-hypocycloid. As a result, we have fully characterised the set of circulant unistochastic matrices for $d\leq 4$, noting their geometric properties (star-shaped with respect to $W_d$) and algebraic properties (monoid for $d\leq 3$ and closedness under multiplication by elements of $\F_4$ for $d=4$).

Let us summarise the known properties of matrices belonging to the sets of bistochastic $\B_d$, bracelet $\L_d$, unistochastic $\U_d$, orthostochastic ${\cal O}_d$ and factorizable $\F_d$ matrices of order $d$  (their definitions are provided in Sec.~\ref{sec:preliminaries}). In the simplest case, $d=2$, all sets do coincide, $\B_2=\L_2=\U_2={\cal O}_2=\F_2$. The situation changes already for $d=3$ as $\B_3 \supset \L_3=\U_3 \supset \F_3$, while orthostochastic matrices form the boundary of the unistochastic set, ${\cal O}_3=\partial \U_3$. For $d=4$ one has $\B_4 \supset \L_4 \supset \U_4$, the sets $\U_4$ and $\F_4 \subset \L_4$ are not comparable and there exist matrices from ${\cal O}_4$ inside the interior of $\U_4$. For a larger dimension $d$ we know that $\B_d \supset \L_d \supset \U_d \supset {\cal O}_d$ and these inclusion relations are proper. Finally, with the additional circulant constraint, we also know that the sets of bracelet and unistochastic matrices coincide for $d=4$, $\C_4\cap\L_4=\C_4\cap\U_4$.

Our investigations brought us one step closer to the full understanding of the set of unistochastic matrices $\U_d$, but also created several questions that are left open for future research. First of all, we found that for $d\leq 4$ every circulant unistochastic matrix can be formed from a circulant unitary matrix, i.e., given a unitary $U$ with circulant transition amplitudes, \mbox{$|U_{j+l,k}|=|U_{j,k-l}|$}, it is always possible to find a circulant unitary $V$ with the same transition amplitudes, $|V_{jk}|=|U_{jk}|$. It thus seems worth checking whether this still holds for higher dimensions. Another observation that holds for $d\leq 4$, and should be further investigated for $d>4$, is whether the bracelet condition is sufficient for unistochasticity if we restrict the analysis to the set of circulant matrices. Furthermore, one can try to verify Conjecture~\ref{conj:bracelet}, stating that $\L_d$ forms a monoid, as it would prove a strong algebraic property of the sets $\U_3$ and $\C_4\cap\U_4$, i.e., they would form monoids. Obviously, the star-shapedness of the set of unistochastic matrices $\U_d$ for $d>3$ is still left open, but one may also want to attack a simpler problem restricting to circulant unistochastic matrices (or even to doubly circulant unistochastic matrices). Finally, further investigations on the structure of factorisable bistochastic matrices $\F_d$ is an interesting problem in itself, independent of unistochasticity.

\bigskip

\textbf{Acknowledgements:} It is a pleasure to thank W.~Bruzda for a helpful discussion on circulant Hadamard matrices. We acknowledge financial support by the Foundation for Polish Science through TEAM-NET project (contract no. POIR.04.04.00-00-17C1/18-00). This research was also supported by the 
National Science Center in Poland under the Maestro grant number
 DEC-2015/18/A/ST2/00274.

\bibliography{bib_channels}

\end{document}